\documentclass[10pt,twocolumn,twoside]{IEEEtran}

\IEEEoverridecommandlockouts

\usepackage{amsmath}
\usepackage{amssymb}
\usepackage{amsfonts}

\usepackage{amsthm}

\usepackage{booktabs}

\usepackage{dsfont}
\usepackage{balance}
\usepackage{graphicx}
\usepackage{comment}
\usepackage{cite}

\usepackage[color=red!50]{todonotes}

\usepackage{algorithmic}
\usepackage{textcomp}

\theoremstyle{plain}

\newtheorem{remark}{Remark}

\newtheorem{lemma}{Lemma}
\newtheorem{corollary}{Corollary}

\theoremstyle{plain}
\newtheorem{theorem}{Theorem}

\theoremstyle{definition}
\newtheorem{definition}{Definition}

\theoremstyle{definition}

\theoremstyle{definition}

\begin{document}

\title{
The Design and Analysis of a Mobility Game
}

\author{Ioannis Vasileios Chremos, \IEEEmembership{Student Member, IEEE}, and Andreas A. Malikopoulos, \IEEEmembership{Senior Member, IEEE}
\thanks{This work was supported by NSF under Grant CNS-2149520.}
\thanks{The authors are with the Department of Mechanical Engineering, University of Delaware, Newark, DE 19716 USA. {\tt\small{\{ichremos,andreas\}@udel.edu.}}}
}

\maketitle

\begin{abstract}

In this paper, we study a routing and travel-mode choice problem for mobility systems with a multimodal transportation network as a ``mobility game" with coupled action sets. We develop a game-theoretic framework to study the impact on efficiency of the travelers' behavioral decision-making. In our framework, we introduce a mobility ``pricing mechanism," in which we model traffic congestion using linear cost functions while also considering the waiting times at different transport hubs. We show that the travelers' selfish actions lead to a pure-strategy Nash equilibrium. We then perform a Price of Anarchy analysis to establish that the mobility system's inefficiencies remain relatively low as the number of travelers increases. We deviate from the standard game-theoretic analysis of decision-making by extending our modeling framework to capture the subjective behavior of travelers using prospect theory. Finally, we provide a simulation study as a proof of concept for our proposed mobility game.

\end{abstract}


%
%
%
%
%

\section{INTRODUCTION}
\label{Section:Introduction}

\subsection{Motivation}

\IEEEPARstart{C}{ommuters} in big cities have continuously experienced the frustration of congestion and traffic jams \cite{Varaiya1993}. Travel delays, accidents, and road altercations have consistently impacted the economy, society, and the natural environment in terms of energy and pollution \cite{colini_baldeschi2017}. One of the pressing challenges of our time is the increasing demand for energy, which requires us to make fundamental transformations in how our societies use and access transportation \cite{zhao2019enhanced}. Thanks to the technological evolution of mobility (e.g., electrification of vehicles, smart mobility with self-driving cars, and improved vehicle sensor technology \cite{Colombo2017}) it is highly expected that we will be able to eliminate congestion while significantly increase mobility efficiency in terms of energy and travel time \cite{sarkar2016}. Several studies have shown the benefits of \emph{emerging mobility systems} (EMS) (e.g., ride-hailing, on-demand mobility services, shared vehicles, self-driving cars) in reducing energy and alleviating traffic congestion in a number of different transportation scenarios \cite{Papageorgiou2002,zardini2020,chalaki2020TCST,Ntousakis:2016aa,Malikopoulos2020}. For a thorough review of models and the possible methods and techniques for smart mobility-on-demand systems see \cite{Zardini2022}.

The cyber-physical nature (e.g., data and shared information) of EMS is associated with significant control challenges and gives rise to a new level of complexity in modeling and control \cite{Ferrara2018}. Research efforts over the last twenty years have tended to focus on the technological dimension. What is missing is a complementary theoretical study to the broader social implications of smart mobility. The impact of selfish social behavior in routing networks of regular and autonomous vehicles has been studied in \cite{Mehr2019,Lazar2021,Biyik2021}. Other efforts have addressed ``how people learn and make routing decisions" with behavioral dynamics \cite{Krichene2015,Krichene2016}. A game-theoretic framework using sequential games was proposed to study the socioeconomic interactions as well as the different tradeoffs that emerge between the mobility stakeholders of a mobility ``ecosystem" \cite{Zardini2021}. It seems though that the problem of how automation in mobility will affect the tendency to travel and decision-making has not been adequately approached yet. In a recent study \cite{Harb2018}, it was shown that when daily commuters were offered a convenient and affordable taxi service for their travels, a change of behavior was noticed; the commuters adjusted their travel behavior and activities and used the taxi service considerably more often leading to an 83\% overall increase in vehicle miles traveled. Along with other similar studies \cite{Auld2017,Bissell2018} this shows that EMS most probably may affect people's tendency to travel and incentivize them to use cars more frequently, which potentially can also lead to a shift away from public transit.

In this paper, we are interested in one open question: \emph{Can we develop an efficient multimodal mobility system that can enhance accessibility while controlling the ratio of travel demand over capacity and improve indirectly the well-being and safety of travelers, passengers, and drivers?} To address this question, we first need to understand the behavioral interactions of travelers with different modes of transportation along with the implications to system efficiency. Thus, we study the game-theoretic interactions of travelers seeking to travel in a transportation network comprised of roads used by different modes of transportation (e.g., cars, buses, light rail, and bikes). A key characteristic of our approach is that we adopt the Mobility-as-a-Service (MaaS) concept, i.e., a multimodal mobility system that handles centrally the travelers' information and provides travel services (e.g., navigation, location, booking, payment). Our goal is to provide a game-theoretic framework that captures the most significant factors of a traveler's decision-making in a transportation network under two different behavioral models.

\subsection{Literature Review}

One of the standard approaches to alleviate congestion in a transportation system has been the management of demand size due to the shortage of space availability and scarce economic resources by imposing an appropriate \emph{congestion pricing scheme} \cite{Pigou2013,Ferguson2021}. Such an approach focuses primarily on intelligent and scalable traffic routing, in which the objective is to optimize the routing decisions in a transportation network and guide and coordinate travelers in choosing these routes \cite{Papageorgiou2003,Bui2012,Zhao2018,Korilis1997,Silva2013}. Game theory has been one of the standard tools that can help us investigate the impact of selfish routing on efficiency and congestion \cite{Marden2015,Hota2020}. By adopting a game-theoretic approach, advanced systems have been proposed to assign travelers concrete routes or minimize all travelers' travel time while studying the system's Nash equilibrium (NE) under different congestion pricing mechanisms \cite{Brown2017,Chandan2019,salazar2019,Chremos2020_ITSC,Chremos2020_CDC,Chremos2021_ECC,Chremos2021_Chapter,Dave2020SocialMedia,Chremos2022_ACC}.

An important and key theoretical approach in alleviating congestion is \emph{routing/congestion games} \cite{Rosenthal1973,Cominetti2009,Hao2018,Lazar2021,Tanaka2021}, which are a generalization of the standard resource-sharing game of an arbitrary number of resources in a network with a finite number of travelers. For example, each traveler may contribute a certain amount of traffic from a source to a destination and affect the overall congestion on a route, thus increasing the travel time for all other travelers. Another important class of games is \emph{potential games}, first introduced in \cite{Monderer1996}, which represent an important branch of game theory. In a potential game the incentive of all players to change their strategy can be expressed using a single global function called the ``potential" function. The potential function depends on the action sets of all players and captures the changes of utility as the actions vary. Potential games have been used extensively in wide-ranging applications; for example, tax schemes of public goods \cite{Konishi1998}, economics of shallow lakes \cite{Maler2003}, electricity markets \cite{Garcia2002}. Routing/congestion and potential games have played an instrumental role in understanding competition over shared resources. Both classes of games have been studied in multiple disciplines to model transportation and communication networks \cite{Orda1993,Roughgarden2015,Sekar2020,Ferguson2021}, common-pool resource games in economics \cite{Ostrom1994}, and resource dilemma problems in psychology \cite{Rapoport1992,Budescu1995}.

In cases where a resource fails (e.g., if a road is viewed as a resource then too many travelers using it may lead to a traffic jam) or there is uncertainty over the resource's quantity/quality, the players of a game cannot collectively reach an efficient equilibrium. Resources with negative congestion externalities have been widely considered in congestion games \cite{Hew2007,Barrera2014}. In this regard, our work deviates from the literature as we consider the overutilization of multiple different resources on each route in the transportation network as well as considering additional indirect costs to the travelers' (e.g., waiting cost at a transport hub). In our modeling framework, we consider \emph{negative congestion externalities} by supposing that if the number of co-travelers that utilize the same route or mode of transportation increases, then a traveler's utility decreases too \cite{Amin2011,Zhong2020}.

So far, most of the existing game-theoretical literature in transportation and routing/congestion games assumes that the players' behavior follows the \emph{rational choice theory}, i.e., each player is a risk-neutral, selfish, utility maximizer \cite{Shoham2008}. This seems to turn most transportation models quite unrealistic, as unexpected travel delays can lead to uncertainty in a traveler's utility. More irrational decision-making over uncertainties and risks in utility can play a significant role, and its study can help us understand how large-scale systems perform inefficiently. There is strong evidence with empirical experiments that show how a human's choices and preferences systematically may deviate from the choice and preferences of a game-theoretic player under the rational choice theory \cite{Kahneman1979,Camerer2004}. This is because real-life decision-making is rarely truly rational, and biases affect how we make decisions. For example, humans compare the outcomes of their choices to a known expected amount of utility (called reference) and decide based on that reference whether their utility is a gain or loss. \emph{Prospect theory} has laid down the theoretical foundations to study such biases and the subjective perception of risk in utility of humans (see the seminal papers \cite{Kahneman1979,Tversky1992}). Prospect theory has been recognized as a closer-to-reality behavioral model for the decision-making of humans and has been used in a wide range of applications and fields \cite{Camerer2004,Barberis2013}, including recent studies in engineering \cite{El_Rahi2016,Nar2017,Etesami2018}. There has also been considerable work at the intersection of transportation studies and prospect theory \cite{de_Palma2006,Lam2001,Etesami2020}.

\subsection{Contribution of this Paper}

In this paper, we propose a game-theoretic framework for the travelers' routing and travel-mode choices in a multimodal transportation network. We study the existence of a NE and the resulting inefficiencies of the travelers' decision-making. One of the most significant aims of our work is to show that although we cannot guarantee equilibrium uniqueness, we can provide an upper bound for the inefficiency that arises from the individual strategic interactions of the travelers. In particular, our modeling framework (called mobility game), considers the impacts of ``negative congestion externalities" and waiting costs in the travelers' decision-making. That way, we offer an improved look at the socioeconomic factors that can affect the efficient and sustainable distribution of travel demand in a transportation network with multiple different modes of transportation (e.g., car, bus, light rail, bike). Moreover, we study the travelers' decision-making under two behavioral models: (1) rational choice theory, where players are selfish and seek to maximize only their own utility; and (2) prospect theory, where the players' biases and subjectivity are taken into account when decisions are made under risk.

The features that distinguish our work from the state of the art are as follows:
    \begin{enumerate}
        \item we model the interactions between travelers using a mobility game, a combination of a routing and a potential game with travel-mode choices and coupled action sets (see Section \ref{Section:Formulation});
        \item we take into account the traffic congestion cost factors using linear cost functions and the waiting time of travelers at different transport hubs; each transport hub allows a traveler to choose any of available modes of transportation to utilize for their travel needs (see Section \ref{Section:Formulation});
        \item we introduce a mobility pricing mechanism to control travel demand and study the inefficiencies at a NE by showing that a NE exists (Theorem \ref{thm:existence_NASH}) and deriving a bound that remains small enough as the number of travelers increases (Theorem \ref{thm:PoA}); and
        \item we incorporate a behavioral model (prospect theory) for the travelers' decision-making under the uncertainty of the transport hubs' budgets (Theorem \ref{thm:existence_NASH_PT}).
    \end{enumerate}

\subsection{Organization of the Paper}

The remainder of the paper is structured as follows. First, in Section \ref{Section:Formulation}, we present the mathematical formulation of our mobility game, which forms the basis of our theoretical study in this paper. Then, in Section \ref{Section:Analysis&Properties}, we derive the properties of our mobility game, i.e., we show NE existence and we bound the Price of Anarchy (PoA) in Subsection \ref{Subsection:NashExistence&PoA}. Then we prove that a NE exists under prospect theory in Subsection \ref{Subsection:NashExistencePT}. We validate our theoretical results with numerical simulations in Section \ref{Section:SimulationResults}. Finally, in Section \ref{Section:Conclusion}, we draw conclusions and offer a discussion of future research.

\section{MODELING FRAMEWORK}
\label{Section:Formulation}

We consider a mobility system of two finite, disjoint, and non-empty sets, (1) the set of travelers $\mathcal{I}$, $|\mathcal{I}| = I \in \mathbb{N}_{\geq 2}$, and (2) the set of mobility services by $\mathcal{J}$, $|\mathcal{J}| = J \in \mathbb{N}$. For example, $j \in \mathcal{J}$ can represent either a car, a bus, a light rail vehicle, or a bike. We consider that in our mobility game, $I < J$. The set of all mobility services $\mathcal{J}$ can be partitioned to a finite number of disjoint subsets, each representing a specific type of a mobility service, i.e., $\mathcal{J} = \bigcup_{h = 1} ^ H \mathcal{J}_h$, where $H \in \mathbb{N}$ is the total number of subsets of $\mathcal{J}$. For example, if there are only two modes of transportation, say cars and buses, then $\mathcal{J} = \mathcal{J}_{1} \cup \mathcal{J}_{2}$, where $\mathcal{J}_1$ represents the subset of all available cars, and $\mathcal{J}_2$ represents the subset of all available buses.

\begin{definition}
    The set of all different types of services is $\mathcal{H} = \{1, \dots, H\}$, $H \in \mathbb{N}$, where each element $h \in \mathcal{H}$ represents a possible travel option. We denote the type of service $j$ used by traveler $i$ by $h_i \in \mathcal{H}$.
\end{definition}

For example, suppose $H = 4$. Then each element $h \in \mathcal{H}$ can be associated one-to-one to the elements of the set $\{ \text{car}, \text{bus}, \text{light rail}, \text{bike} \}$.

Naturally, each service can accommodate up to a some finite number of travelers that is different for each type of services. So, we expect the ``physical traveler capacity" of each service to vary significantly. 

\begin{definition}
    Each service $j \in \mathcal{J}$ is characterized by a current \emph{physical traveler capacity}, i.e., $\varepsilon_j \in \{ 0, 1, 2, \dots, \bar{\varepsilon}_j \}$, where $\bar{\varepsilon}_j \in \mathbb{N}$ denotes the maximum traveler capacity of service $j$.
\end{definition}

For example, one bus can provide travel services up to eighty travelers (seated and standing) compared to a bike-sharing company's bike (one bike per traveler).

Travelers seek to travel in a transportation network represented by a directed multigraph $\mathcal{G} = (\mathcal{V}, \mathcal{E})$, where each node in $\mathcal{V}$ represents a city area (neighborhood) with a ``transport hub," i.e., a central place where travelers can use different modes of transportation. Each edge $\mathcal{E}$ represents a sequence of city roads with public transit lanes. For our purposes, we think of $\mathcal{G} = (\mathcal{V}, \mathcal{E})$ as a representation of a smart city network with a road and public transit infrastructure. In network $\mathcal{G}$, any traveler $i \in \mathcal{I}$ seeks to travel from an origin $o \in \mathcal{V}$ to a destination $d \in \mathcal{V}$ while making optional stops at a self-chosen transport hub $v_i \in \mathcal{V}$. So, on one hand, all travelers are associated with the same origin-destination pair $(o, d)$. On the other hand, travelers can make a stop along their route. Next, each type of mobility services $h \in \mathcal{H}$ is associated with a sequence of edges, i.e., a route that connects at least two nodes (or transport hubs). We say that there exists a set of routes for each traveler $i$ where each route connects their origin-destination pair $(o, d)$ and can be traveled by any mobility service. Formally, we have $\mathcal{P} ^ {(o, d)} \subset 2 ^ {\mathcal{E}}$ to denote the set of routes available to traveler $i$ in origin-destination $(o, d)$, where each route in $\mathcal{P} ^ {(o, d)}$ consists of a set of edges connecting $o$ to $d$.

Each traveler $i \in \mathcal{I}$ seeks to travel in network $\mathcal{G}$ using one of the available mobility services $j \in \mathcal{J}$ of type $h \in \mathcal{H}$. Thus, any traveler can choose the type of mobility service they prefer for their specific travel needs. This means that travelers compete with each other for the available services in the transportation network. For the purposes of this work, we restrict our attention to all available modes of transportation that use the road infrastructure. In addition, each transport hub (including the ones at $(o, d)$ allows travel by any mobility service (any mode of transportation), thus a traveler's travel preferences or needs can be satisfied by many and different mobility services (as one expects from a multimodal transportation network).

Selfish behavior, however, may lead to inefficiencies. Therefore, as part of our efforts to \emph{control} the inefficiencies that arise from the travelers' selfishness (and thus control the emergence of rebound effects), under our modeling framework we introduce the idea of a ``mobility pricing mechanism" to incentivize travelers to use services in public transit for their travel needs. Informally, each transport hub starts with a budget, collects payments for services, and then provides monetary incentives (pricing mechanism) to travelers to ensure a \emph{socially-efficient} utilization of services in the network. By ``socially-efficient" we mean that the endmost collective travel outcomes must achieve two objectives: (i) respect and satisfy the travelers' preferences regarding mobility, and (ii) ensure the alleviation of congestion in the system. We formalize the idea of our mobility pricing mechanism in the following definitions.

\begin{definition}\label{defn:tokens}
    Each traveler $i$ starts with a \emph{mobility wallet} represented in monetary units by $\theta_i \in [0, \bar{\theta}_i]$, where $\bar{\theta}_i \in  \mathbb{R}_{> 0}$ is the maximum amount of traveler $i$'s monetary units. Traveler $i$ uses their wallet $\theta_i$ to pay for their travels in network $\mathcal{G}$.
\end{definition}

\begin{definition}
    Any traveler $i$ is required to pay a ticket, called ``mobility payment," for using a mobility service in network $\mathcal{G}$. This mobility payment is given by some function $\pi_i : \mathcal{H} \times \mathbb{N} \to \mathbb{R}_{\geq 0}$, where $0 \leq \pi_i(h_i, \varepsilon_j) \leq \bar{\theta}_i$.
\end{definition}

Note that $\pi_i(h_i, \varepsilon_j)$ has the same monetary units as $\theta_i$ in Definition \ref{defn:tokens}. Intuitively, a traveler $i$ pays $\pi_i$ for using mobility service $j$ of type $h_i$. The mobility payment $\pi_i$ of any traveler $i$ varies extensively for each type of service $h_i$ and increases fast as $\varepsilon_j$ tends to $\bar{\varepsilon}_j$. To ensure our exposition is compact, we omit the arguments of $\pi_i(h_i, \varepsilon_j)$, and simply write $\pi_i$.

In our modeling framework,  each traveler $i$ pays for using a mobility service $j$ of type $h_i$ on route $\rho_i$ with origin-destination pair $(o, d)$ making an optional stop at transport hub $v_i$. At each transport hub, available funds can be offered to incentivize travelers to use public transit services. Our mobility game can be thought of as a static game that is played repeatedly \cite{Webster2014}, thus travelers are assumed to take different actions multiple times. Therefore, the pricing mechanism needs to consider both the payments of all travelers for each type of service and the available funds at each transport hub. We formalize this idea for the allocation of mobility payments for each traveler $i$ by stating the following definition.

\begin{definition}\label{defn:set_transport_hub}
    Suppose traveler $i \in \mathcal{I}$ chooses route $\rho_i \in \mathcal{P} ^ {(o, d)}$ and makes a stop at transport hub $v_i$ along that route using some service $j \in \mathcal{J}$ of type $h_i \in \mathcal{H}$. Then, the set of co-travelers at $v_i \in \mathcal{V}$ is $\mathcal{S}_{v_i} = \{k \in \mathcal{I} \; | \; v_k = v_i\}$.
\end{definition}

In words, $\mathcal{S}_{v_i}$ groups all travelers including traveler $i$ who have made a stop at transport hub $v_i$. Next, we formally define the available budget at transport hub $v_i$.

\begin{definition}\label{defn:hub_tokens}
    Let $b(v_i) \in \mathbb{R}$ be the amount of funds available for transactions at traveler $i$'s transport hub $v_i \in \mathcal{V} \setminus \{(o, d)\}$ over all types of services $h \in \mathcal{H}$.
\end{definition}

Intuitively, $b(v_i)$ represents the available funds (e.g., after covering all expenses), in the same monetary units as $\theta_i$ and $\pi_i$, that transport hub $v_i$ may allocate to travelers. Practically, even though our proposed mobility game is not dynamic, $b(v_i)$ can be computed based on historical data (e.g., along similar lines presented in \cite{Bandi2014}), and thus capture the ``demand" of services at a particular $v_i$. Each traveler $i \in \mathcal{I}$ starts with a mobility wallet $\theta_i$ and pays $\pi_i$ while they can make a stop at transport hub $v_i$. The latter allows us to model transfers.

We can capture the travelers' preferences of different outcomes using a ``utility function." Travelers are expected to act as utility maximizers. Thus, we can influence the travelers' behavior by introducing a \emph{control input} in the utility function. In our modeling framework, we consider a \emph{mobility pricing mechanism}, as a control input, that aims to reward or penalize each traveler $i$ (either by increasing a traveler's utility or decreasing it). We offer here an informal description of our pricing mechanism. The total excess amount of mobility funds is $b(v_i) - \sum_{i \in \mathcal{S}_{v_i}} \pi_i$. The total excess amount of mobility funds at transport hub $v_i$ excluding traveler $i$ is $b(v_i) - \sum_{k \in \mathcal{S}_{v_i} \setminus \{i\}} \pi_k$. Given the available mobility funds already present at $v_i$, we can redistribute the ``mobility wealth" based on the types of services and roads utilized by the travelers as follows: we consider a quadratic-based pricing mechanism $\tau(v_i, \pi_i)$, defined formally next, which is the same for all travelers.
Under this pricing mechanism, we observe the following two interesting properties: for high values of $\pi_i$, $\tau$ is strictly decreasing; for low values of $\pi_i$, $\tau$ is strictly increasing. Thus, if traveler $i$ pays a high payment $\pi_i$ (e.g., which implies traveler $i$ uses a car), then disincentive is also high to use this mobility service. Thus, this serves as an indirect incentive for a traveler to use public transit or a different transport hub (this becomes clear in \eqref{eqn:utility}). Furthermore, the pricing mechanism $\tau$ can take negative values, and actually strictly decreases fast as $\pi_i$ takes high values for any traveler $i$. So, travelers can get penalized if they choose a ``high-demand" type of service (thus, leading to a high valued $\sum_{k \in \mathcal{S}_{v_i} \setminus \{i\}} \pi_k$). Even if a traveler has the means to pay (i.e., the traveler has a large mobility wallet), the pricing mechanism can penalize the traveler with hefty fees, thus all travelers have the incentive to minimize the penalties and choose public transit services or less congested transport hubs. For example, when a traveler uses a bike, their mobility payment will be low and so they can earn (instead of paying for the service). This incentivizes a sustainable allocation of services to all travelers.
We offer the formal definition of the pricing mechanism next for the allocation of the mobility funds and payments.

\begin{definition}\label{defn:pricing_mechanism}
    The pricing mechanism is a multivariable function $\tau \mapsto \mathbb{R}$ that depends on a traveler $i$'s transport hub $v_i$ and mobility payment $\pi_i$, and is explicitly given by
        \begin{multline}\label{eqn:pricing}
            \tau(v_i, \pi_i) = \\
            \left( b(v_i) - \sum_{k \in \mathcal{S}_{v_i} \setminus \{i\}} \pi_k \right) ^ 2 - \left( b(v_i) - \sum_{i \in \mathcal{S}_{v_i}} \pi_i \right) ^ 2.
        \end{multline}
\end{definition}

Recall that the term $b(v_i)$ captures the demand of a transport hub $v_i$ based on what types of services in general travelers have been using (e.g., if a transport hub has a lot of money, it means travelers use cars significantly).

\begin{remark}
    If we expand \eqref{eqn:pricing} and simplify, we obtain the following relation
        \begin{equation}
            \tau(v_i, \pi_i) = 2 \pi_i \left( b(v_i) - \frac{\pi_i}{2} - \sum_{k \in \mathcal{S}_{v_i} \setminus \{i\}} \pi_k \right).
        \end{equation}
    The behavior of \eqref{eqn:pricing} is rather interesting. Obviously, as traveler $i$'s payment increases, then \eqref{eqn:pricing} decreases. However, given that $b(v_i) > \sum_{k \in \mathcal{S}_{v_i} \setminus \{i\}} \pi_k$, for small values of $\pi_i$, $\tau$ increases up to a maximum point, and then starts to decrease. This characteristic of \eqref{eqn:pricing} serves as a strong incentive for travelers to choose services that are ``cheap" (bikes) or uncongested (buses) since then $\pi_i$ will be small. Otherwise, $\tau$ can take very high negative values as $\pi_i$ increases.
\end{remark}

As long as $b(v_i)$ is higher than $\sum_{k \in \mathcal{S}_{v_i} \setminus \{i\}} \pi_k$, then the pricing mechanism \eqref{eqn:pricing} redistributes wealth back to each traveler $i$ based on what is available on the self-chosen transport hub $v_i$ and how much travelers pay by taking into consideration traveler $i$'s contribution at transport hub $v_i$.

Since the travelers' objective is to maximize their payoff, we need a way to ``incentivize" travelers to avoid decisions that may lead to an empty mobility wallet. Thus, we introduce an empty wallet ``disincentive" for an arbitrary traveler $i$.

\begin{definition}\label{defn:risk}
    Given the current amount of mobility wallet $\theta_i$ of any traveler $i$, the \emph{disincentive} of having an empty wallet is a decreasing function $g : [0, \bar{\theta}_i] \to \mathbb{R}$ given by
        \begin{equation}\label{eqn:token_risk}
            g(\theta_i) = \frac{\bar{\theta}_i}{\theta_i + \eta_i \pi_i},
        \end{equation}
    where $\eta_i \in (0, 1)$ is a socioeconomic characteristic of traveler $i$ and affects the impact of how much they choose to spend or save in terms of their mobility wallet.
\end{definition}

Definition \ref{defn:risk} establishes mathematically a disincentive as a function where $\bar{\theta}_i$ is proportional to the sum of the current mobility wallet $\theta_i$ and the weighted mobility payment $\eta_i \pi_i$. We expect each traveler to avoid as much as possible an empty wallet; hence, \eqref{eqn:token_risk} ensures to ``penalize" travelers with a low mobility wallet $\theta_i$ while choosing to spend $\eta_i \pi_i \approx \theta_i$. Thus, \eqref{eqn:token_risk} grows fast as $\theta_i$ decreases.
We offer now the formal definition of a traveler's action set.

\begin{definition}\label{defn:action_set}
    For an arbitrary traveler $i \in \mathcal{I}$, the action set is $\mathcal{A}_i = \mathcal{P} ^ {(o, d)} \times \mathcal{V} \times \mathbb{R}_{\geq 0}$, where $\mathcal{P} ^ {(o, d)}$ is the set of routes that connects traveler $i$'s origin-destination pair $(o, d)$, $\mathcal{V}$ is the set of nodes in network $\mathcal{G}$ that includes all possible transports hubs $v_i$, and $\pi_i \in \mathbb{R}_{\geq 0}$ is the mobility payment of traveler $i$.
\end{definition}

Note that, by Definition \ref{defn:action_set}, the action set $\mathcal{A}_i$ of an arbitrary traveler $i$ is a coupled set with discrete values (route, transport hub, source/destination pair), and continuous values (mobility payment). Thus, the action profile $a_i \in \mathcal{A}_i$ of traveler $i\in\mathcal{I}$ is a vector of discrete and continuous values. We write $\mathcal{A} = \mathcal{A}_1 \times \mathcal{A}_2 \times \dots \times \mathcal{A}_I$ for the Cartesian product of all the travelers' action sets. We write $a_{- i} = (a_1, a_2, \dots, a_{i - 1}, a_{i + 1}, \dots, a_I)$ for the action profile that excludes traveler $i \in \mathcal{I}$. Next, for the aggregate action profile, we write $a = (a_i, a_{- i})$, $a \in \mathcal{A}$. We also denote by $a ^ *$, $a ^ {\text{Nash}}$ an action profile at a social optimum and at a NE, respectively.

Next, we introduce a travel time latency function to capture the congestion cost that travelers may experience.

\begin{definition}\label{defn:linear_time_latency}
    Let the total number of services of all types $h = 1, \dots, H$ on road $e \in \mathcal{E}$ with at least one traveler be $J_e = \sum_{h \in \mathcal{H}} \omega_h |\mathcal{J}_{e, h}|$, where $\mathcal{J}_{h, e}$ is the set of all services on road $e$ of type $h$, and $(\omega_h)_{h \in \mathcal{H}}$, $\omega_h \in [0, 1]$ are weight parameters that depend on the type $h$ to capture the different impact of services on the traffic. Then, the travel time latency function is a strictly increasing linear function $c_e : \mathbb{N} \to \mathbb{R}$, with explicit form $c_e(J_e) = \xi_1 J_e + \xi_2$, where $\xi_1, \xi_2$ are constants.
\end{definition}

Notice that we assume linearity in the travel time latency functions $c_e$, which is not unique in the literature \cite{Roughgarden2002,Roughgarden2007,Etesami2020}. The justification behind linearity is that it is the simplest yet most useful way for a mathematical analysis to capture the travel costs in terms of distance or road capacity and the traffic congestion costs. The choice of the constants $\xi_1$ and $\xi_2$ play an important role, namely $\xi_2$ can represent the length of road $e$ and $\xi_1$ normalizes the number of services on road $e$ so that both components of $c_e$ have the same units.

We can now formally define the utility of any traveler $i\in\mathcal{I}$.

\begin{definition}
    The utility $u_i : \mathcal{A} \to \mathbb{R}$ of traveler $i\in\mathcal{I}$ is what traveler $i$ receives under the risk-neutral setting given by
        \begin{multline}\label{eqn:utility}
            u_i(a) = \tau(v_i, \pi_i) - g(\theta_i) \\
            - \zeta_1 \left( \sum_{e \in \rho_i : \rho_i \in \mathcal{P} ^ {(o, d)}} c_e(J_e) \right) - \zeta_2 \left( \frac{|\mathcal{S}_{v_i}|}{\sigma(v_i, h_i)} \right),
        \end{multline}
    where $\sigma(v_i, h_i)$ is the rate of travel service at transport hub $v_i$ for type of service $h_i$ (how many travelers per hour can travel using type of service $h_i$ from transport hub $v_i$), and $\zeta_1, \zeta_2 \in \mathbb{R}$ are unit parameters that transform time to money (that way the units of \eqref{eqn:utility} are consistent).
\end{definition}

Note that both constants $\zeta_1, \zeta_2$ get absorbed by the constants of function $c_e$ (as defined in Definition \ref{defn:linear_time_latency}) and parameter $\sigma$, respectively. So, we can safely omit them from the mathematical analysis. In \eqref{eqn:utility}, the first term represents the pricing mechanism and is the amount of mobility funds redistributed to traveler $i$ for paying $\pi_i$. The second term is the disincentive as defined in Definition \ref{defn:risk}, and the third term is the congestion cost of traveler $i$ due to traffic on road $e$. Finally, the last term in \eqref{eqn:utility} is a waiting cost for joining transport hub $v_i$, where the number of travelers at transport hub $v_i$ is proportional to the rate of travel service at transport hub $v_i$.

Next, we characterize the mobility game in our modeling framework.

\begin{definition}\label{defn:mobility_game}
    The mobility game is fully characterized by the tuple $\mathcal{M} = \langle \mathcal{I}, \mathcal{J}, (\mathcal{A}_i)_{i \in \mathcal{I}}, (u_i)_{i \in \mathcal{I}} \rangle$, a collection of sets of travelers, mobility services, actions, and a profile of utilities.
\end{definition}

The mobility game $\mathcal{M}$ is a non-cooperative repeated routing game with a multimodal transportation network and coupled action sets. The travelers have a travel-mode choice to make that will satisfy their travel needs. The benefit of ensuring that our mobility game $\mathcal{M}$ is a repeated game is that it eliminates the possibility of unassigned travelers. At this point also, we clarify the \emph{information structure} of the mobility game $\mathcal{M}$ (``who knows what?" \cite{Malikopoulos2021}). All travelers have complete knowledge of the mobility system (network, travel time latencies and waiting costs, and utility functions). Each traveler knows their own information (action and utility) as well as the information of other travelers. For the purposes of our work, we observe that a NE is most fitting to apply as a solution concept as it requires complete information.

\section{ANALYSIS AND PROPERTIES}
\label{Section:Analysis&Properties}

In this section, our goal is to establish the existence of at least one NE in the mobility game $\mathcal{M}$, derive an upper bound for the PoA, and perform a prospect theory analysis. 

We start our exposition by providing a summary of two necessary preliminary concepts and results of game theory that we use throughout the paper.

\begin{definition}\label{defn:potential}
    A game $\mathcal{M}$ is an \emph{exact potential game} if there exists a potential function $\Phi : \mathcal{A} \to \mathbb{R}$ such that for all $i \in \mathcal{I}$, for all $a_{- i}$, and for all $a_i, a_i '\in\mathcal{A}_i$, we have
        \begin{equation}
            \Phi(a_i, a_{- i}) - \Phi(a_i ', a_{- i}) = u_i(a_i, a_{- i}) - u_i(a_i ', a_{- i}).
        \end{equation}
\end{definition}


\begin{definition}\label{defn:NE}
    An action profile $a ^ {\text{Nash}} = (a_i ^ {\text{Nash}}, a_{- i} ^ {\text{Nash}})$ is called a pure-strategy \emph{Nash equilibrium} for game $\mathcal{M}$ if, for all $i \in \mathcal{I}$, we have $u_i(a_i ^ {\text{Nash}}, a_{- i} ^ {\text{Nash}}) \geq u_i(a_i ', a_{- i} ^ {\text{Nash}})$, for all $a_i ' \in \mathcal{A}_i$.
\end{definition}

The potential function $\Phi$ is a useful tool in showing whether a game has a NE and analyzing its properties. This is because by construction the effect on utility of any traveler's action is expressed by one function common for all travelers.

\subsection{Existence of a Nash Equilibrium}
\label{Subsection:NashExistence&PoA}

In this subsection, we prove that for the mobility game $\mathcal{M}$, as defined in Definition \ref{defn:mobility_game}, there exists at least one NE. The key idea of our proof is the use of a potential function, as defined in Definition \ref{defn:potential}, that captures the changes in utility of an arbitrary traveler that deviates in their action.

\begin{theorem}\label{thm:existence_NASH}
    The mobility game $\mathcal{M}$ admits a pure-strategy NE.
\end{theorem}

\begin{proof}
    See Appendix \ref{appendix_1}.
\end{proof}

Note that this existence result is not straightforward as the action set of any traveler is a coupled set composed of countable and uncountable subsets.

\begin{corollary}\label{cor:convergence}
    If the mobility game $\mathcal{M}$ is played repeatedly, then the travelers' actions converge to a pure-strategy NE in finite time.
\end{corollary}

\begin{proof}
    This is a consequence of Theorem \ref{thm:existence_NASH} and follows from Theorem 2.6 (pp. 33) in \cite{La2016}. It is sufficient to note that the mobility game $\mathcal{M}$ is a repeated routing game with complete information. So, at first, traveler $i \in \mathcal{I}$ chooses their action $a_i$ according to the existing information of the mobility system (travel time latencies, network congestion). Then any other traveler $k \in \mathcal{I}$ chooses their action $a_k$ based on what traveler $i$ has chosen to do and the information state of the mobility system is updated (who uses what route and what transport hub is busy and how many funds are collected at each transport hub). In turn, traveler $i$ can improve their action $a_i$ to $a_i '$, which gives a chance to any other traveler $k$ to improve too. Thus, since all travelers compete against each other for the best possible utility, the mobility game $\mathcal{M}$ is guaranteed to reach at least one NE.
\end{proof}

\subsection{Price of Anarchy Analysis}

An existence result (Theorem \ref{thm:existence_NASH}) leads to the problem of multiple NE and raises questions to the efficiency of each equilibrium. For example, an important concern is the efficiency of the equilibrium that the travelers will reach (as it is guaranteed by Corollary \ref{cor:convergence}). To address this concern, we provide an analysis based on the \emph{Price of Anarchy} (PoA) \cite{Koutsoupias1999}, which is one of the most widely-used metrics to measure the inefficiency in a system and provides an understanding of how the travelers' decision-making affect the overall performance of the system. We provide the formal definition of the PoA.
    
\begin{definition}
    Let the \emph{social welfare} of the mobility game $\mathcal{M}$ be represented by $J(a) = \sum_{i \in \mathcal{I}} u_i(a)$. Then, the PoA is the ratio of the maximum optimal social welfare over the minimum social welfare at a NE, i.e.,
        \begin{equation}\label{defn:poa}
            \text{PoA} = \frac{\max_{a \in \mathcal{A}} \sum_{i \in \mathcal{I}} u_i(a)}{\min_{a \in \mathcal{A} ^ {\text{Nash}}} \sum_{i \in \mathcal{I}} u_i(a)} \geq 1,
        \end{equation}
    where $\mathcal{A} ^ {\text{Nash}}$ is the set of NE, which is guaranteed to be non-empty according to Theorem \ref{thm:existence_NASH}.
\end{definition}

Next, we show that, for the mobility game $\mathcal{M}$, \eqref{defn:poa} is as low as possible at an arbitrary NE. Thus, it follows that our PoA result yields an upper bound for the inefficiencies at a NE of the mobility game $\mathcal{M}$.

\begin{theorem}\label{thm:PoA}
    Any inefficiencies of any NE of the mobility game $\mathcal{M}$ remain low as close to a constant as the number of travelers $|\mathcal{I}| = I$ tends to infinity. Mathematically, we have
        \begin{equation}\label{eqn:thm_PoA_relation}
            \text{PoA} \leq 2 + \frac{5}{I} \sum_{v \in \mathcal{V}} \left( \sum_{h \in \mathcal{H}} b(v, h) \right) ^ 2.
        \end{equation}
\end{theorem}

\begin{proof}
    See Appendix \ref{appendix_4}.
\end{proof}

We discuss now the intuition behind Theorem \ref{thm:PoA}. The travelers of the mobility game $\mathcal{M}$ are considered selfish and non-cooperative, thus one important question is what will be the impact of selfishness in the efficiency of the mobility system. Since the existence of a NE is guaranteed by Theorem \ref{thm:existence_NASH} and the mobility game $\mathcal{M}$ converges to at least one NE, we can compare the level of inefficiency at a NE to the social optimum (this is exactly what the PoA does). The bound we have derived in \eqref{eqn:thm_PoA_relation} under certain conditions ensures the mobility system's inefficiency is guaranteed to remain within a constant and as the number of travelers increases this bound becomes smaller and smaller. So, our modeling framework can ensure in a realistic setting (big city with road infrastructure) with a large number of travelers a sufficiently efficient operation of the mobility system as it could ideally be operated by a central authority ``ordering" the travelers how to travel. Our bound in \eqref{eqn:thm_PoA_relation} is strong in the sense that it excludes any other possibility of an improvement in efficiency compared to what we can achieve at a NE. Furthermore, our mobility game $\mathcal{M}$ is a special case of the model considered in \cite{Awerbuch2013,Roughgarden2011}, in which the authors show that the bound is exact and tight.

\subsection{Prospect Theory Analysis}
\label{Subsection:NashExistencePT}

In this subsection, we provide an introduction to prospect theory and its main concepts \cite{Wakker2010,Barberis2013}. We then incorporate prospect theory to our modeling framework. One of the main questions prospect theory attempts to answer is how a decision-maker may evaluate different possible actions/outcomes under uncertain and risky circumstances. Thus, prospect theory is a descriptive behavioral model and focuses on three main behavioral factors:
    \begin{enumerate}
        \item \emph{Reference dependence}: decision makers make decisions based on their utility, which is measured from the ``gains" or ``losses." However, the utility is a gain or loss relative to a reference point that may be unique to each decision maker. It has been shown in experimental studies \cite{Barberis2013}, the reference dependence captures the tendency of a decision-maker to be affected in their decisions by the \emph{changes in attributes} than the \emph{absolute magnitudes}. For example, shortest/average travel time between two locations.
        \item \emph{Diminishing sensitivity}: changes in value have a greater impact near the reference point than away from the reference point. For example, an individual is highly likely to discriminate between a 1 and 2 hours travel time, but not very likely to notice the difference between 18 and 19 hours travel time.
        \item \emph{Loss aversion}: decision makers are more conservative in gains and more risky in losses. For example, a traveler may prefer to secure a 45 min commute rather than risking for a 1.5 hours commute.
    \end{enumerate}
One way to mathematize the above behavioral factors (1) - (3), is to consider an action by a decision-maker as a ``gamble" with objective utility value $z \in \mathbb{R}$ (e.g., money). We say that this decision maker \emph{perceives} $z$ subjectively using a \emph{value function} \cite{Tversky1992,Al-Nowaihi2008}
    \begin{equation}\label{eqn:value_function}
        \nu(z) =
            \begin{cases}
                (z - z_0) ^ {\beta_1}, & \text{if } z \geq z_0, \\
                - \lambda (z_0 - z) ^ {\beta_2}, & \text{if } z < z_0,
            \end{cases}
    \end{equation}
where $z_0$ represents a reference point, $\beta_1, \beta_2 \in (0, 1)$ are parameters that represent the diminishing sensitivity. Both $\beta_1, \beta_2$ shape \eqref{eqn:value_function} in a way that the changes in value have a greater impact near the reference point than away from the reference point.
We observe that \eqref{eqn:value_function} is concave in the domain of gains and convex in the domain of losses. Moreover, $\lambda \geq 1$ reflects the level of loss aversion of decision makers (see Fig. \ref{fig:prospect1}).

\begin{remark}
    To the best of our knowledge, there does not exists a widely-agreed theory that determines and defines the reference dependence \cite{Kahneman1979,Koszegi2007,Baucells2011}. In engineering \cite{Hota2016,Etesami2020}, it is assumed that $z_0 = 0$ capturing a decision maker's expected status-quo level of the resources.
\end{remark}

As we discussed earlier in this subsection, prospect theory models the subjective behavior of decision makers under uncertainty and risk. Each objective utility $z \in \mathbb{R}$ is associated with a probabilistic occurrence, say $p \in [0, 1]$. Decision makers though are subjective and perceive $p$ in different ways depending on its value. To capture this behavior, we introduce a strictly increasing function $w : [0, 1] \to \mathbb{R}$ with $w(0) = 0$ and $w(1) = 1$ called the \emph{probability weighting function}. This function allows us to model how decision makers may overestimate small probabilities of objective utilities, i.e., $w(p) > p$ if $p$ is close to $0$, or underestimate high probabilities, i.e., $w(p) < p$ if $p$ is close to $1$ (see Fig. \ref{fig:prospect2}). For the purposes of this work, we use the probability weighting function first introduced in \cite{Prelec1998},
    \begin{equation}\label{eqn:prelec_weight}
        w(p) = \exp \left(- (- \log (p)) ^ {\beta_3} \right), \quad p \in [0, 1],
    \end{equation}
where $\beta_3 \in (0, 1)$ represents a \emph{rational index}, i.e., the distortion of a decision-maker's probability perceptions. Mathematically, $\beta_3$ controls the curvature of the weighting function (see Fig. \ref{fig:prospect2}). Although there are many different formulations for the probability weighting function, we use \eqref{eqn:prelec_weight} defined in \cite{Prelec1998} as it is a single-parameter function and it can be computed in polynomial time.

\begin{figure}[ht]
    \centering
    \includegraphics[width = 1.0 \columnwidth]{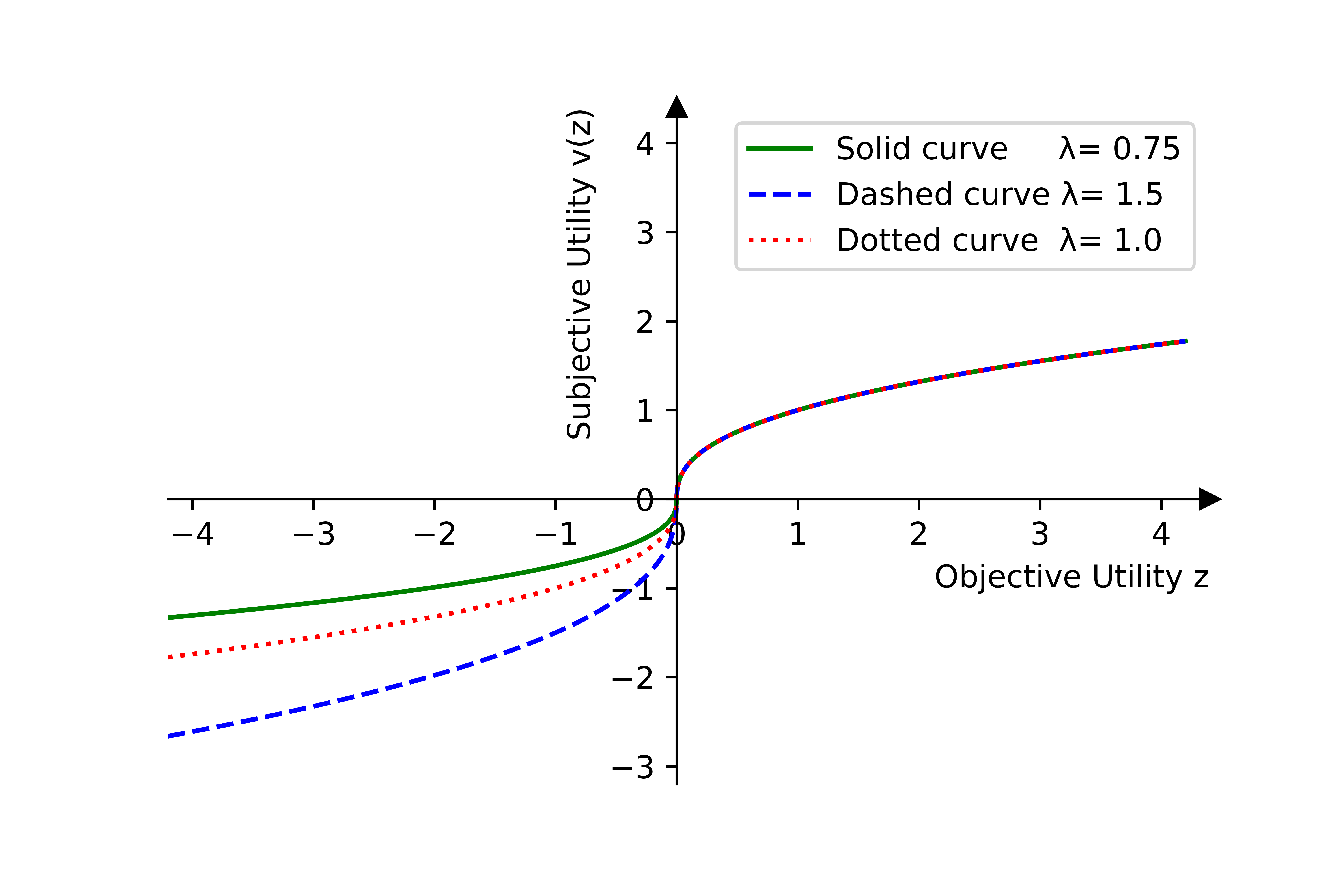}
    \caption{The value function for three different values of $\lambda$.}
    \label{fig:prospect1}
\end{figure}

\begin{figure}[ht]
    \centering
    \includegraphics[width = 1.0 \columnwidth]{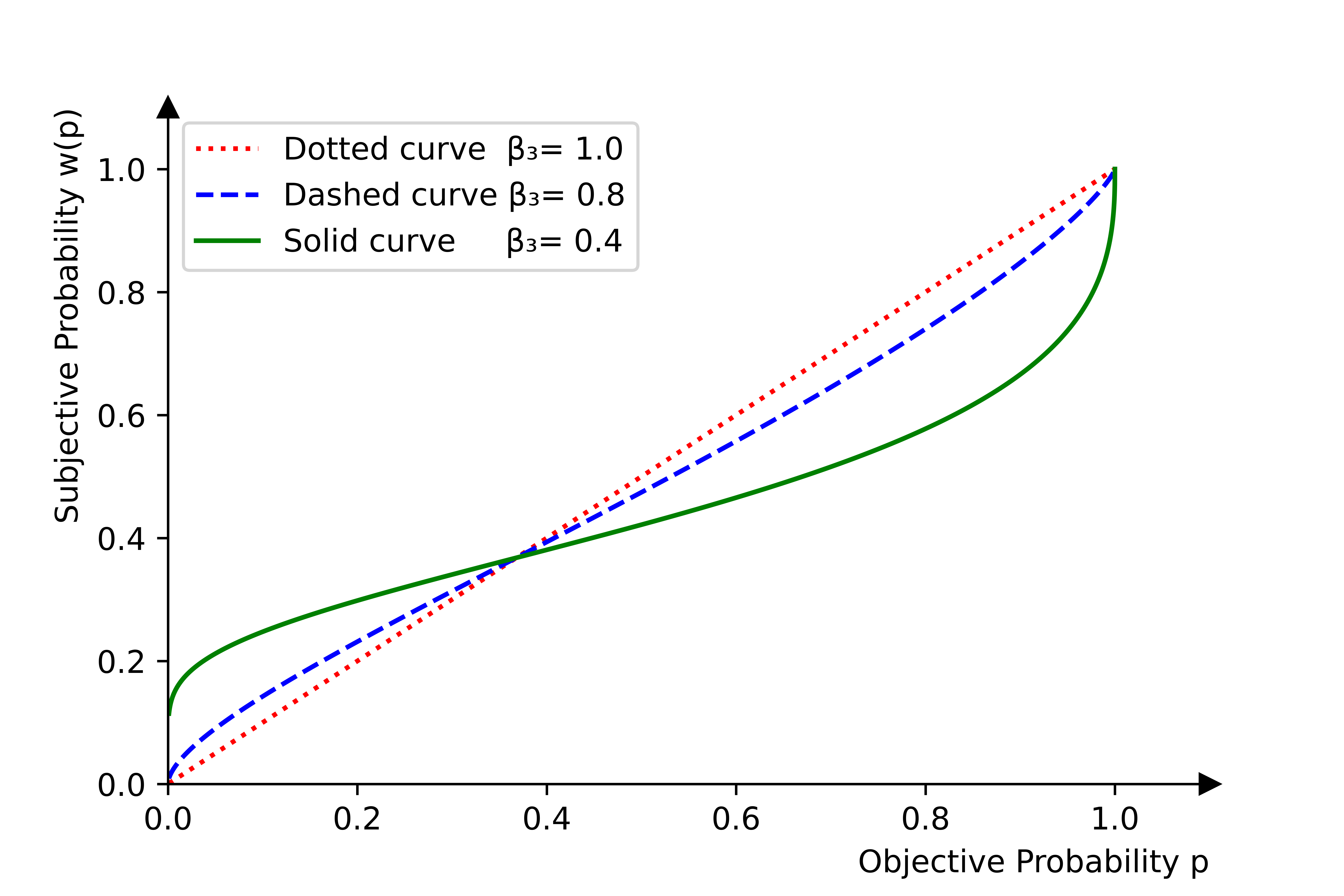}
    \caption{The probability weighting function for three different values of the rational index $\beta_3$.}
    \label{fig:prospect2}
\end{figure}

Next, we define a \emph{prospect} which is a tuple of the objective utility (gain or loss) and its probability of happening.

\begin{definition}
    Suppose that there are $K \in \mathbb{N}$ possible outcomes available to a decision-maker and $z_k \in \mathbb{R}$ is the $k$th gain/loss of objective utility. Then a prospect $\ell_k$ is a tuple of the utilities and their respective probabilities $\ell_k = (z_0, z_1, z_2, \dots, z_K; p_0, p_1, p_2, \dots, p_K)$, where $k = 1, 2, \dots, K$. We denote the $k$th prospect more compactly as $\ell_k = (z_k, p_k)$. We have that $\sum_{k = 0} ^ K p_k = 1$ and $\ell_k$ is well-ordered, i.e., $z_0 \leq x_1 \leq \cdots \leq z_K$. Under prospect theory, the decision-maker evaluates their ``subjective utility" as $u(\ell) = \sum_{0 \leq k \leq K} v(z_k) w(p_k)$, where $\ell = (\ell_k)_{k = 1} ^ K$ is the profile of prospects of $K$ outcomes.
\end{definition}

In the remainder of this subsection, we apply the prospect theory to our modeling framework, clearly define the mobility outcomes (objective and subjective utilities), and then show that the prospect-theoretic mobility game $\mathcal{M}$ admits a NE.

Travelers may be uncertain on the available amount of mobility funds at any transport hub, that is why we we define a \emph{mobility prospect} to represent as a random variable $Z$ with objective utilities $z_1, z_2, \dots, z_K$ and their probabilities $p_1, p_2, \dots, p_K$. Each $z_k$ now represents the uncertain $b(v_i)$. In addition, the reference dependence of each traveler $i$ is represented by $z_i ^ 0 \in \mathbb{R}$. For any traveler $i$, the probability weighting function is $w_i : [0, 1] \to \mathbb{R}$ and the value function is $\nu_i(z_k, z_i ^ 0) : \mathbb{R} ^ 2 \to \mathbb{R}$, $k = 1, 2, \dots, K$. Thus, we have
    \begin{equation}\label{eqn:expected_utility_PT}
        \mathbb{E} [Z] = \sum_{k = 1} ^ K \nu_i(z_k, z_i ^ 0) w_i(p_k),
    \end{equation}
where $w_i(p_k)$ is given by \eqref{eqn:prelec_weight}, and
    \begin{equation}\label{eqn:valuation_prospect}
        \nu_i(z_k, z_i ^ 0) =
            \begin{cases}
                (z_k - z_i ^ 0) ^ {\beta}, & \text{if } z_k \geq z_i ^ 0, \\
                - \lambda (z_i ^ 0 - z_k) ^ {\beta}, & \text{if } z_k < z_i ^ 0,
            \end{cases}
    \end{equation}
where $\beta = \beta_1 = \beta_2$. We can justify $\beta_1 = \beta_2$ in the above definition as it has been verified to produce extremely good results and the outcomes are consistent with the original data \cite{Tversky1992}.
Next, we explicitly define the reference point for the mobility game $\mathcal{M}$ as follows $z_i ^ 0 = \left( \sum_{k \in \mathcal{S}_{v_i} \setminus \{i\}} \pi_k \right) ^ 2 - \left( \sum_{i \in \mathcal{S}_{v_i}} \pi_i \right) ^ 2$ where $z_i ^ 0$ represents the ideal redistribution of wealth to traveler $i$ (since no transport hub $v_i$ should make a profit, i.e., $b(v_i) = 0$). For the random variable $Z$, we assume a continuous distribution $F$ with zero mean and a probability density function $f$, and explicitly have
    \begin{equation}
        Z = \left( F - \sum_{k \in \mathcal{S}_{v_i} \setminus \{i\}} \pi_k \right) ^ 2 - \left( F - \sum_{i \in \mathcal{S}_{v_i}} \pi_i \right) ^ 2.
    \end{equation}
So, by \eqref{eqn:expected_utility_PT}, we have $\mathbb{E} [Z] = \sum_{n \in \mathbb{R}} \nu_i(z(n), z_i ^ 0) w_i(f(n))$, where $z(n)$ represents at each transport hub $v_i$ of an arbitrary traveler $i$ the realization of $Z$ with $n \in \mathbb{R}$ available mobility funds. The total utility now under prospect theory for a traveler $i$ is
    \begin{multline}\label{eqn:utility_prospect}
        u_i ^ {\text{PT}}(a) = z_i ^ 0 + \mathbb{E} [Z] - \frac{\bar{\theta}_i}{\theta_i + \eta_i \pi_i} \\
        - \sum_{e \in \rho_i : \rho_i \in \mathcal{P} ^ {(o, d)}} c_e(J_e) - \frac{|\mathcal{S}_{v_i}|}{\sigma(v_i, h_i)} .
    \end{multline}

Next, we show that our mobility game $\mathcal{M}$ under prospect theory is guaranteed to have at least one NE.

\begin{theorem}\label{thm:existence_NASH_PT}
    The mobility game $\mathcal{M}$ under prospect theory admits a pure-strategy NE.
\end{theorem}

\begin{proof}
    See Appendix \ref{appendix_5}.
\end{proof}

\begin{corollary}\label{cor:PT}
    For the mobility game $\mathcal{M}$ under prospect theory, the sequence of best responses of an arbitrary traveler $i \in \mathcal{I}$ converges to a NE.
\end{corollary}

\begin{proof}
    It is sufficient to note that the action set $\mathcal{A}_i$ of any traveler $i$ is compact, thus, it follows from the results in \cite{Monderer1996} that the sequence of best responses of any traveler $i \in \mathcal{I}$ converges to a NE.
\end{proof}

Both Theorem \ref{thm:existence_NASH_PT} and Corollary \ref{cor:PT} ensure that the mobility game $\mathcal{M}$ under the prospect-theoretic behavioral model admits a NE and prospect-based travelers will eventually converge to it. Both results establish that we can still ensure that an equilibrium can be reached under certain conditions for the cost and pricing functions.

\section{SIMULATION RESULTS}
\label{Section:SimulationResults}

In this section, we conduct a simulation study to demonstrate the theoretical results of our proposed game-theoretic framework. We consider a simple traffic network (see Fig. \ref{fig:network}) with 5 directed roads $\{e_1, e_2, e_3, e_4, e_5\}$, 3 transport hubs $\{O, A, B\}$, and a final destination transport hub $D$. Note that there are three possible routes $\{e1-e4, e2-e5, e2-e3-e4\}$. As it is standard, we consider two independent and identically distributed Gaussian distributions $F \sim N(0, 10)$ and $Q_i \sim N(0, 1)$ for the budget at different transport hubs and for the socioeconomic characteristic $\eta_i$, respectively. Also, we consider that all the travelers have identical wallets with $\bar{\theta}_i = 10$, and $\theta_i = 2$ for each $i$ traveling from the origin $O$ to the destination $D$ (see Fig. \ref{fig:network}).

\begin{figure}[ht]
    \centering
    \includegraphics[width = 1.0 \columnwidth]{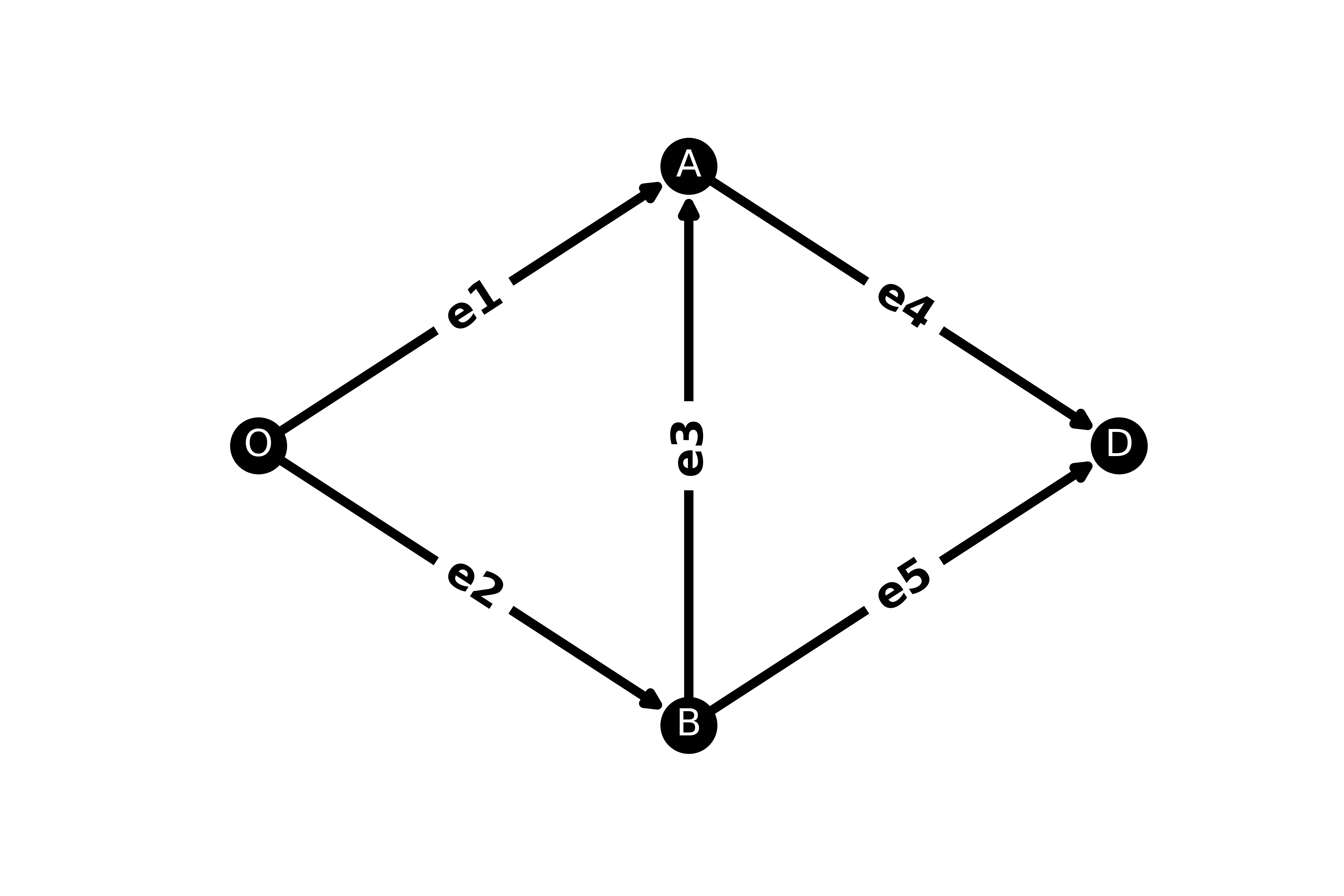}
    \caption{A visualization of a simple transportation network $\mathcal{G} = (\mathcal{V}, \mathcal{E})$.}
    \label{fig:network}
\end{figure}

\subsection{PoA under Rational Choice Theory}

\begin{figure}[ht]
    \centering
    \includegraphics[width = 1 \columnwidth]{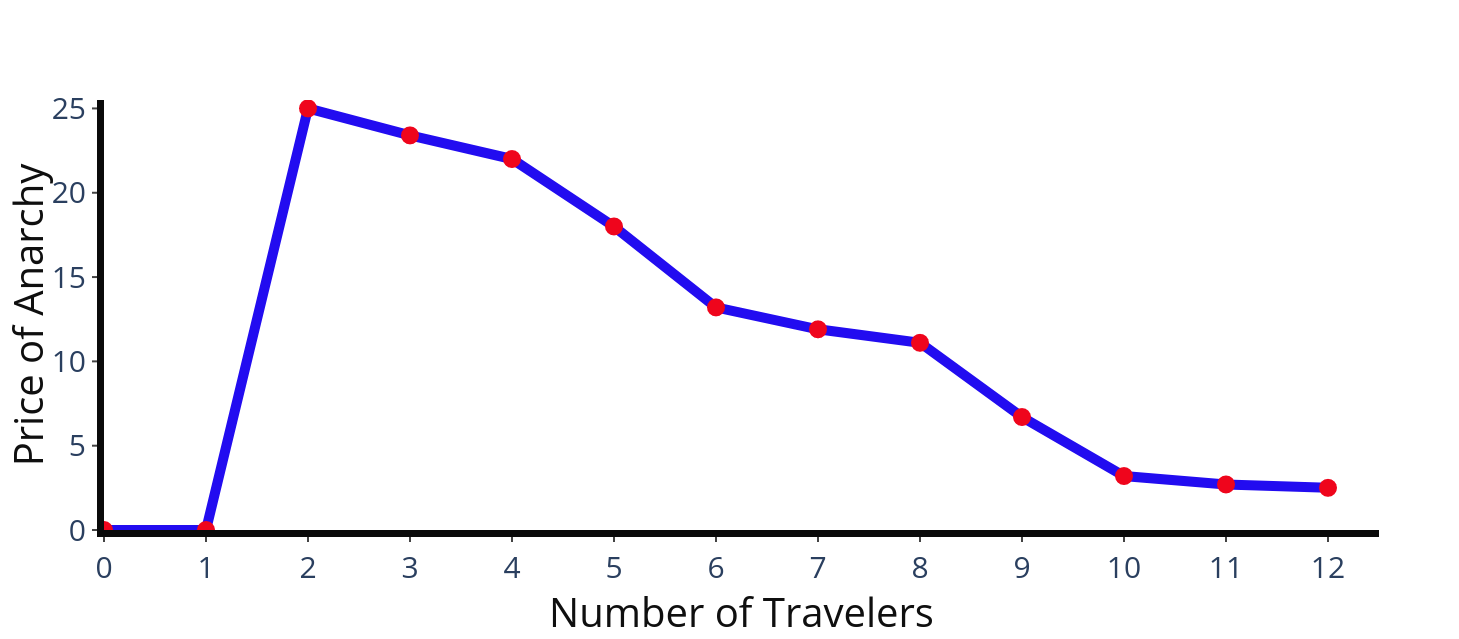}
    \caption{Price of Anarchy computed using \eqref{defn:poa} with respect to the number of travelers.}
    \label{fig:simple_case}
\end{figure}

\begin{figure}[ht]
    \centering
    \includegraphics[width = 1 \columnwidth]{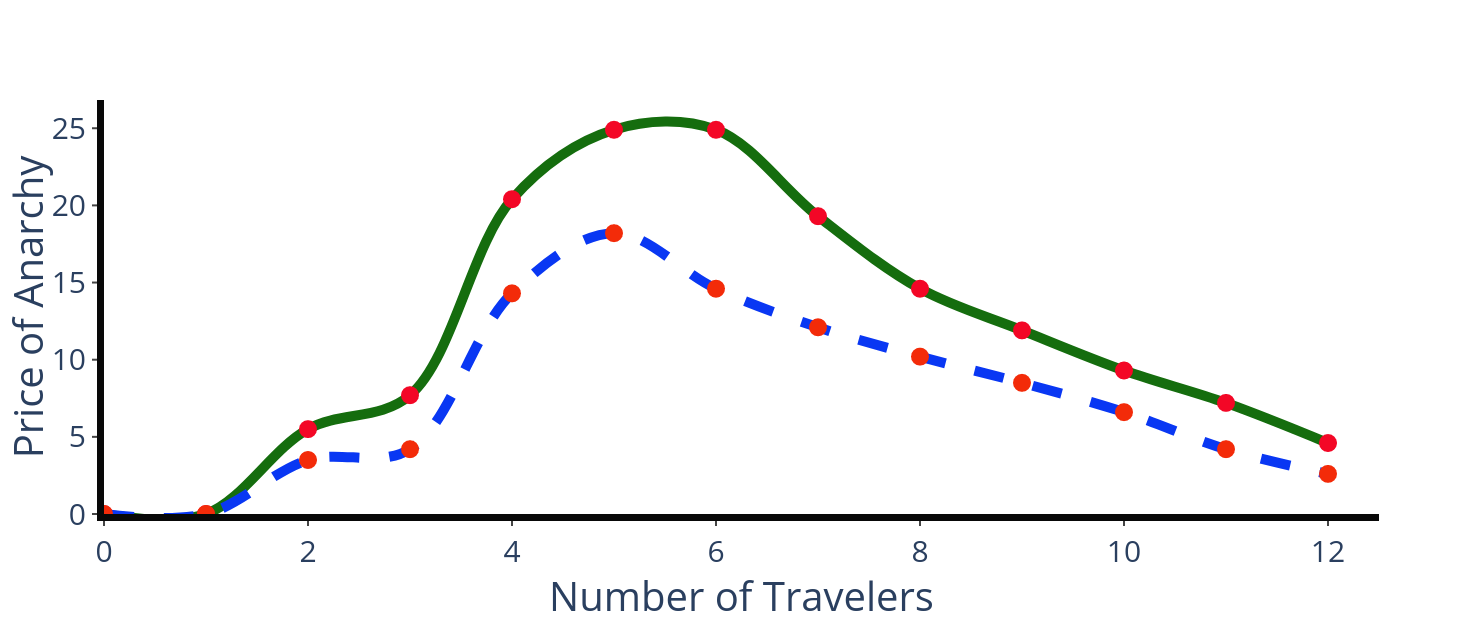}
    \caption{Price of Anarchy computed using \eqref{defn:poa} with respect to the number of travelers under different pricing mechanisms: $\tau ^ {(1)}$ and $\tau ^ {(2)}$.}
    \label{fig:simple_case_Comparison}
\end{figure}

In Fig. \ref{fig:simple_case}, we illustrate how the PoA under the rational choice theory deviates as more travelers (up to $I$ = 12) join the mobility system. In Fig. \ref{fig:simple_case_Comparison}, we compare the outcomes for different choices of pricing and latency functions. Here, we let the number of travelers increase from $I = 2$ to $I = 12$, and then we compute the PoA with the same linear travel time latency function $c_e(J_e) = 5 J_e + 3$ but with different pricing functions: (i) $\tau ^ {(1)} (v_i, \pi_i) = \left( b(v_i) - \sum_{k \in \mathcal{S}_{v_i} \setminus \{i\}} \pi_k \right) ^ 2 - \left( b(v_i) - \sum_{i \in \mathcal{S}_{v_i}} \pi_i \right) ^ 2$, and (ii) $\tau ^ {(2)} (v_i, \pi_i) = \left( b(v_i) - \sum_{k \in \mathcal{S}_{v_i} \setminus \{i\}} \pi_k \right) ^ {2 / 3} - \left( b(v_i) - \sum_{i \in \mathcal{S}_{v_i}} \pi_i \right) ^ {2 / 3}$. We observe that the non-linearity of the pricing functions worsens the PoA. However, we can still see that as the number of travelers increases, efficiency improves. This verifies our Theorem \ref{thm:PoA}.

\begin{remark}\label{rmk:simulation}
    Unfortunately, the promise of small inefficiency as the number of travelers increases (Theorem \ref{thm:PoA}) cannot be guaranteed if we consider a quadratic travel time latency function or a highly non-linear pricing function that exhibits increased non-linearity.
\end{remark}

Given Remark \ref{rmk:simulation}, our modeling framework can be implemented in real-world mobility systems as long as we can control the nonlinearity in the travel time latencies and the pricing functions. Since nonlinearity in the functions can been seen to affect the PoA negatively by a considerable margin (see Fig. \ref{fig:simple_case_Comparison}), it is most important how we model the travel time latency and the pricing functions.

For the network illustrated in Fig. \ref{fig:network} with $I = 12$ travelers, we summarize one instance of the mobility game $\mathcal{M}$ with the NE strategies of each traveler in Table \ref{tab:numerical}. We have assumed at random that each transport hub starts with a budget of $b(A) = 9.2$, and $b(B) = 15.1$ for transport hubs $A$, and $B$, respectively. Table \ref{tab:numerical} can be read as follows: a NE strategy for traveler $i = 1$ is from $O$, take the route $\rho_1 = \{e1 - e4\}$ using a bike with a mobility price of $\tau = 1$, i.e., they receive a payment in the amount of \$1.

\begin{table}[ht]\label{tab:numerical}
    \caption{Mobility game $\mathcal{M}$ at NE with 12 travelers, 3 possible routes $\{\rho_1 = e1-e4, \rho_2 = e2-e5, \rho_3 = e2-e3-e4\}$, 3 transport hubs $\{O, A, B\}$, and a final destination hub $D$.}
    \centering
        \begin{tabular}{c|c|c|c|c|c}
            $i$ & $h_i$ & $\rho_i ^ 1$ & $\rho_i ^ 2$ & $\rho_i ^ 3$ & $\tau$ \\ \midrule
            1 & bike & 1 & 0 & 0 & 1.0 \\
            2 & bike & 1 & 0 & 0 & 1.0 \\
            3 & bike & 1 & 0 & 0 & 1.0 \\
            4 & bus & 0 & 0 & 1 & 0.75 \\
            5 & bus & 0 & 0 & 1 & 0.75 \\
            6 & bus & 0 & 0 & 1 & 0.75 \\
            7 & car & 0 & 1 & 0 & -2.5 \\
            8 & car & 0 & 1 & 0 & -3.5 \\
            9 & car & 0 & 1 & 0 & -3.5 \\
            10 & car & 0 & 1 & 0 & -3.5 \\
            11 & car & 0 & 1 & 0 & -5.25 \\
            12 & car & 0 & 1 & 0 & -5.25 \\ \midrule
        \end{tabular}
\end{table}

\subsection{PoA under Prospect Theory}

In this subsection, we perform a PoA simulation study under prospect theory. We consider the same network as before (see Fig. \ref{fig:network}) and $I = 12$ travelers. For simplicity, all travelers have the same weighting and valuation functions given by \eqref{eqn:prelec_weight} and \eqref{eqn:valuation_prospect}, respectively, where $(\beta_1, \beta_2, \lambda) = (0.88, 0.88, 2.25)$. All three specific values are the ones recommended by experimental studies on human subjects \cite{Booij2010}. In Fig. \ref{fig:prospect}, we have compared the PoA to the rational index $\beta_3$ and observe the somewhat random perturbations of the PoA depending on $\beta_3$. Since $\beta_3$ represents the distortion of a decision-maker's probability perceptions (where 0 means irrationality), it is expected to see great deviations in the efficiency of the mobility system. When travelers are considered neither rational or irrational ($\beta_3 \approx 0.5$) we may get more inefficient NE even though the social optimum is evaluated at maximum rationality ($\beta_3 = 1$).

\begin{figure}[h]
    \centering
    \includegraphics[width = 1 \columnwidth]{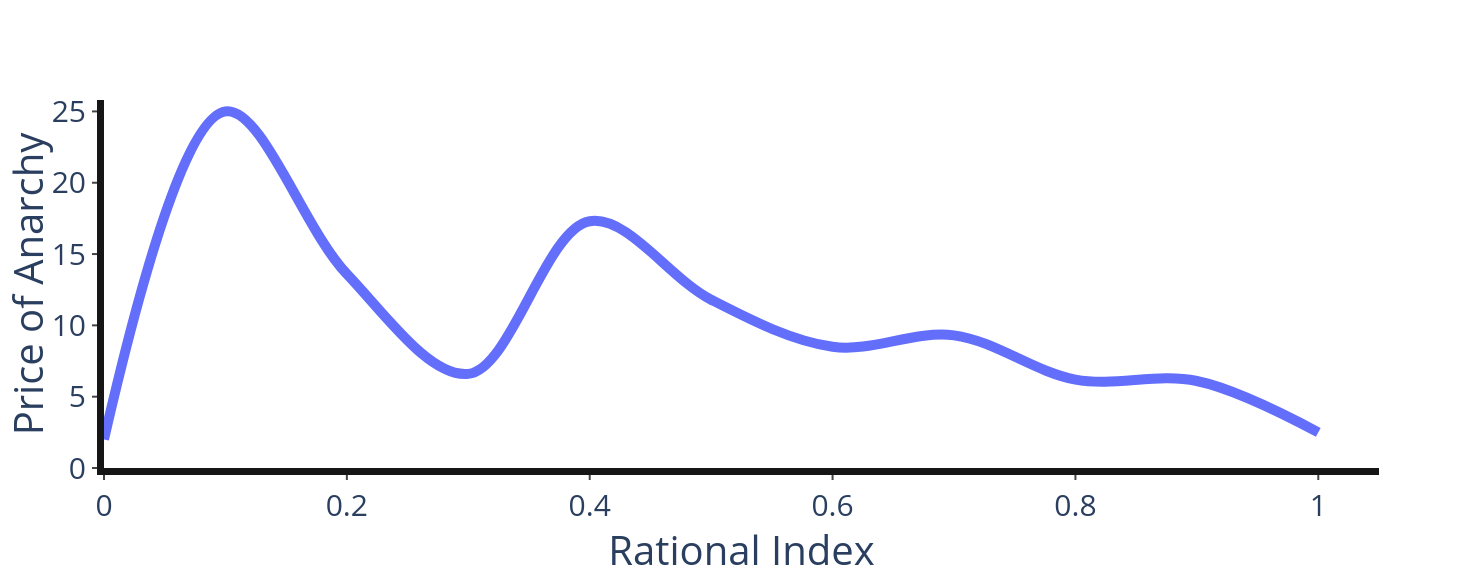}
    \caption{PoA with respect to the \emph{rational index} $\beta_3$ under prospect theory.}
    \label{fig:prospect}
\end{figure}

\section{CONCLUDING REMARKS}
\label{Section:Conclusion}

In this paper, we proposed a game-theoretic framework (called mobility game) to study the behavioral interactions of travelers in a multimodal transportation network. First, we formulated a repeated non-cooperative routing game with a finite number of travelers. In our first result, we showed that the mobility game admits a NE under the rational choice theory. In our second result, we derived a bound for the PoA. Although we cannot have uniqueness at an equilibrium, our upper bound guarantees that the inefficiencies is as low as possible if the number of travelers is large enough (which is naturally expected in a mobility system). Next, we extended our framework to consider the subjective behavior of travelers under prospect theory, and showed that our mobility game admits a NE. Finally, we performed a simulation study under different mobility scenarios and behavioral models and showed how the PoA deviates as the number of travelers increases.

\subsection{Implementation}

In this subsection, we outline how our proposed framework can be potentially implemented. We consider a major metropolitan area with an extensive road and public transit infrastructure; a good example is Boston. Several key areas in Boston are connected by roads, buses, light rail, and bikes. These areas can serve as transport hubs from which travelers can utilize any of the available modes of transportation. We can apply the MaaS concept and offer on each transport hub travel services (e.g., navigation, location, booking, payment) to all passing travelers. Information can be shared among all travelers via a ``mobility app," which allows travelers to access the services on the transport hubs. Using this app, travelers will be able to pay for their travel needs and, at the same time, receive mobility payments. For example, a traveler who informs the app and uses a bike multiple times (per day or per week) can receive mobility payments. In addition, travelers travel multiple times and interact with each other more than once. So, travelers seek to move from one place to another while competing with many other travelers, use the transport hubs to access their preferred mode of transportation, and pay using a mobility app. Each mode of transportation offers different benefits in utility; for example, a car is more convenient than a bus and is expected to be in high demand. This naturally will lead to inefficiency and congestion.

The mobility game $\mathcal{M}$ with the utility structure as defined in \eqref{eqn:utility} captures the key factors that may play a role in a traveler's decision-making. It can be seen by Theorem \ref{thm:existence_NASH} and Corollary \ref{cor:convergence} that an equilibrium exists and can be reached by the travelers without direct intervention from a central authority. The particular pricing mechanism we have proposed in \eqref{eqn:pricing} ensures all travelers with the computational power of their cellphone can quickly derive the NE strategy (route, transport hub, payment). This is important as we can avoid solving a mixed-integer nonlinear program for all the travelers in the mobility system. In addition, by Theorem \ref{thm:PoA} we can guarantee that the inefficiency of the mobility system stays low as long as the number of travelers remains large (something that is expected in a typical mobility system). We can see from the simulation results (Figs. \ref{fig:simple_case} and \ref{fig:simple_case_Comparison}) that under different pricing mechanisms (different quadratic functions), the PoA decreases by a considerable margin as we increase the number of travelers. Thus, even though we cannot guarantee the uniqueness of a NE, we can ensure that all NE are similarly efficient and nearly as efficient as the social optimum as long as the number of travelers is high.

Using prospect theory, we can also consider how a traveler can feel uncertain about whether they may receive mobility payments for choosing a more sustainable mode of transportation (e.g., bike). Under certain conditions, we show that indeed a NE exists (Theorem \ref{thm:existence_NASH_PT}) and it can be reached by the travelers as they can travel from the hub that is nearest to their home to the hub that is nearest to their work (Corollary \ref{cor:PT}). Thus, our modeling framework is proved to lead to a NE under two different behavioral models and capture the impact of the travelers' decision-making.

\subsection{Limitations and Future Work}

One important limitation of our framework is the assumption of complete information. Realistically, we cannot expect travelers to have accurate and complete knowledge of other travelers or the system's capabilities (network, road capacities). A potential direction for future research should relax this assumption by only allowing travelers to know their own actions and utilities. This is a standard technique in Bayesian game-theoretic analyses. Another interesting direction for future research is to expand the current framework by explicitly designing the socially-efficient pricing functions to achieve the best possible equilibrium in the mobility system using techniques from mechanism design. Furthermore, to showcase the benefits of the proposed game-theoretic modeling framework a necessary extension of our work is using machine learning techniques \cite{Balcan2005,Qian2021,Vineyard2012,Lim2020,Kim2020} with real-life data.

\section*{Acknowledgment}
The authors would like to thank the anonymous reviewers for their constructive and valuable comments and suggestions that helped improve the quality and presentation of the paper.

\appendices

\section{Proof of Theorem \ref{thm:existence_NASH}}
\label{appendix_1}

\begin{proof}
    As it is standard in existence of NE results for potential games (see Chapter 2 \cite{La2016}), we start by stating the explicit form of the potential function for the mobility game $\mathcal{M}$, i.e.,
        \begin{multline}\label{eqn:thm1-first}
            \Phi(a) = \sum_{v \in \mathcal{V}} \left( b(v) - \sum_{i \in \mathcal{S}_{v}} \pi_i \right) ^ 2 - \sum_{i \in \mathcal{I}} \frac{\bar{\theta}_i}{\theta_i + \eta_i \pi_i} \\
            - \sum_{e \in \mathcal{E}} \sum_{k = 1} ^ {J_e} c_e(k) - \sum_{v \in \mathcal{V}} \sum_{i \in \mathcal{I}} \frac{|\mathcal{S}_{v}| (|\mathcal{S}_{v}| + 1)}{2 \sigma(v, h_i)},
        \end{multline}
    Our goal now is to verify Definition \ref{defn:potential}, thus show that $\mathcal{M}$ is a potential game. Mathematically, for an arbitrary traveler $i$ and for any two actions $a_i = (\rho_i \in \mathcal{P} ^ {(o, d)}, v_i, \pi_i)$ and $a_i ' = (\rho_i ' \in \mathcal{P} ^ {(o, d)}, v_i ', \pi_i ')$, we need to show $\Phi(a_i, a_{- i}) - \Phi(a_i ', a_{- i}) = u_i(a_i, a_{- i}) - u_i(a_i ', a_{- i})$.
    
    Note here that any traveler $i$ that deviates in their action $a_i$ to $a_i '$ by changing their route $\rho_i$ to $\rho_i '$ that does not require an additional service $j$ on any new roads $e \in \rho_i '$, then traveler $i$'s impact to congestion is negligent. Thus, traveler $i$ can change routes and still travel along an existing service $j$ in road $e \in \rho_i '$ if that service $j$ has not reached its maximum capacity $\bar{\varepsilon}_j$. If traveler $i$ changes their route from $\rho_i$ to $\rho_i '$ and the traveler capacities of all services on that route are not maxed out, then the number of services $J_e$ does not change in the roads that remain the same along both routes (any $e \in \rho_i \cap \rho_i '$). However, the number of services $J_e$ increases by one in any roads $e \in \rho_i ' \setminus \rho_i$ since we require an additional service $j$ in road $e$ for traveler $i$. Thus, we have $\sum_{e \in \rho_i '} c_e(J_e) = \sum_{e \in \rho_i ' \cap \rho_i} c_e(J_e) + \sum_{e \in \rho_i ' \setminus \rho_i} c_e(J_e + 1)$.
    
    If traveler $i$ changes their transport hub $v_i$ to $v_i '$, then their new waiting cost is $\frac{|\mathcal{S}_{v_i '}| + 1}{\sigma(v_i ', h)}$, where $\mathcal{S}_{v_i '} = \{k \in \mathcal{I} \setminus \{i\} \; | \; v_k = v_i '\}$. Along the same lines for $\pi_i$ and $\pi_i '$ and we can write
        \begin{multline}\label{eqn:eqn:thm1-second}
            u_i(a_i ', a_{- i}) =  \left( b(v_i ') - \sum_{k \in \mathcal{S}_{v_i '} \setminus \{i\}} \pi_k \right) ^ 2 \\
            - \left( b(v_i ') - \pi_i ' - \sum_{k \in \mathcal{S}_{v_i '} \setminus \{i\}} \pi_k \right) ^ 2 - \frac{\bar{\theta}_i}{\theta_i + \eta_i \pi_i '} - \frac{|\mathcal{S}_{v_i '}|}{\sigma(v_i ', h_i)} \\
            - \left[ \sum_{e \in \rho_i ' \cap \rho_i} c_e(J_e) + \sum_{e \in \rho_i ' \setminus \rho_i} c_e(J_e + 1) \right].
        \end{multline}
    We subtract \eqref{eqn:eqn:thm1-second} from \eqref{eqn:utility} to get
        \begin{multline}
            u_i(a_i, a_{- i}) - u_i(a_i ', a_{- i}) = \\
            \left( b(v_i) - \sum_{k \in \mathcal{S}_{v_i} \setminus \{i\}} \pi_k \right) ^ 2 - \left( b(v_i) - \sum_{i \in \mathcal{S}_{v_i}} \pi_i \right) ^ 2 \\
            - \left( b(v_i ') - \sum_{k \in \mathcal{S}_{v_i '} \setminus \{i\}} \pi_k \right) ^ 2 + \left( b(v_i ') - \pi_i ' - \sum_{k \in \mathcal{S}_{v_i '} \setminus \{i\}} \pi_k \right) ^ 2 \\
            - \sum_{e \in \rho_i \setminus \rho_i '} c_e(J_e) + \sum_{e \in \rho_i ' \setminus \rho_i} c_e(J_e + 1) \\
            - \frac{|\mathcal{S}_{v_i}|}{\sigma(v_i, h_i)} + \frac{|\mathcal{S}_{v_i '}| + 1}{\sigma(v_i ', h_i)} + \frac{\eta_i \bar{\theta}_i (\pi_i - \pi_i ')}{(\theta_i + \eta_i \pi_i)(\theta_i + \eta_i \pi_i ')},
        \end{multline}
    where $\sum_{e \in \rho_i \setminus \rho_i '} c_e(J_e) = \sum_{e \in \rho_i} c_e(J_e) - \sum_{e \in \rho_i ' \cap \rho_i} c_e(J_e)$.
    
    Now, we denote all four components of \eqref{eqn:thm1-first} as follows: $\phi_1 = \sum_{v \in \mathcal{V}} \left( b(v) - \sum_{i \in \mathcal{S}_{v}} \pi_i \right) ^ 2$, $\phi_2 = - \sum_{i \in \mathcal{I}} \frac{\bar{\theta}_i}{\theta_i + \eta_i \pi_i}$, $\phi_3 = - \sum_{e \in \mathcal{E}} \sum_{k = 1} ^ {J_e} c_e(k)$, and $\phi_4 = - \sum_{v \in \mathcal{V}} \sum_{i \in \mathcal{I}} \frac{|\mathcal{S}_{v}| (|\mathcal{S}_{v}| + 1)}{2 \sigma(v, h_i)}$. Step by step, we compute the difference of all four different $\phi$'s under $a_i$ and $a_i '$ as follows
        \begin{multline}\label{eqn:thm1-2nd}
            \phi_1(a_i, a_{- i}) - \phi_1(a_i ', a_{- i}) =  \left( b(v_i) - \sum_{k \in \mathcal{S}_{v_i} \setminus \{i\}} \pi_k \right) ^ 2 \\
            - \left( b(v_i) - \sum_{i \in \mathcal{S}_{v_i}} \pi_i \right) ^ 2 - \left( b(v_i ') - \sum_{k \in \mathcal{S}_{v_i '} \setminus \{i\}} \pi_k \right) ^ 2 \\
            + \left( b(v_i ') - \pi_i ' - \sum_{k \in \mathcal{S}_{v_i '} \setminus \{i\}} \pi_k \right) ^ 2,
        \end{multline}
        \begin{multline}\label{eqn:thm1-3rd}
            \phi_2(a_i, a_{- i}) - \phi_2(a_i ', a_{- i}) = \\
            - \sum_{i \in \mathcal{I}} \frac{\bar{\theta}_i}{\theta_i + \eta_i \pi_i} + \left[ \sum_{k \in \mathcal{I} \setminus \{i\}} \frac{\bar{\theta}_k}{\theta_k + \eta_k \pi_k} + \frac{\bar{\theta}_i}{\theta_i + \eta_i \pi_i '} \right] \\
            = \sum_{i \in \mathcal{I}} \frac{\eta_i \bar{\theta}_i (\pi_i - \pi_i ')}{(\theta_i + \eta_i \pi_i)(\theta_i + \eta_i \pi_i ')}.
        \end{multline}
        \begin{multline}\label{eqn:thm1-4th}
            \phi_3(a_i, a_{- i}) - \phi_3(a_i ', a_{- i}) = \\
            - \sum_{e \in \mathcal{E}} \sum_{k = 1} ^ {J_e} c_e(k) + \sum_{e \in \mathcal{E} \setminus \{e \in \rho_i ' \}} \sum_{k = 1} ^ {J_e - 1} c_e(k) + \sum_{e \in \mathcal{E} \setminus \{e \in \rho_i\}} \sum_{k = 1} ^ {J_e + 1} c_e(J_e) \\
            = - \sum_{e \in \rho_i \setminus \rho_i '} c_e(J_e) + \sum_{e \in \rho_i ' \setminus \rho_i} c_e(J_e + 1),
        \end{multline}
        \begin{multline}\label{eqn:thm1-5th}
            \phi_4(a_i, a_{- i}) - \phi_4(a_i ', a_{- i}) = \sum_{v \in \mathcal{V}} \sum_{i \in \mathcal{I}} \frac{|\mathcal{S}_{v}| (|\mathcal{S}_{v}| + 1)}{2 \sigma(v, h_i)} \\
            - \sum_{v \in \mathcal{V} \setminus \{v_i\} \cup \{v_i '\}} \sum_{k \in \mathcal{I} \setminus \{i\}} \left[ \frac{|\mathcal{S}_{v}| (|\mathcal{S}_{v}| + 1)}{2 \sigma(v, h_k)} \right] - \frac{|\mathcal{S}_{v_i}| (|\mathcal{S}_{v_i}| - 1)}{2 \sigma(v_i, h_i)}
            \\ - \frac{(|\mathcal{S}_{v_i '}| + 1) (|\mathcal{S}_{v_i '}| + 2)}{2 \sigma(v_i ', h_i)} = \frac{|\mathcal{S}_{v_i}|}{\sigma(v_i, h_i)} - \frac{|\mathcal{S}_{v_i '}| + 1}{\sigma(v_i ', h_i)},
        \end{multline}
    We define $\Phi(a_i, a_{- i}) - \Phi(a_i ', a_{- i}) = \sum_{k = 1} ^ 4 \left[ \phi_k(a_i, a_{- i}) - \phi_k(a_i ', a_{- i}) \right]$. We take the sum of \eqref{eqn:thm1-2nd} - \eqref{eqn:thm1-5th} and thus, we obtain $\Phi(a_i, a_{- i}) - \Phi(a_i ', a_{- i}) = u_i(a_i, a_{- i}) - u_i(a_i ', a_{- i})$. This proves that the mobility game $\mathcal{M}$ is a potential game and so following from key results (see \cite{La2016}) we conclude that $\mathcal{M}$ admits a pure-strategy NE.
\end{proof}

\section{Proof of Lemma \ref{lemma1}}
\label{appendix_2}

\begin{lemma}\label{lemma1}
    Let $a_i ^ * = (\rho_i ^ *, v_i ^ *, \pi_i ^ *)$ denote the optimal action of traveler $i \in \mathcal{I}$, define $\tilde{b} ^ 2 = \sum_{v \in \mathcal{V}} \left( \sum_{h \in \mathcal{H}} b(v, h) \right) ^ 2 = \sum_{v \in \mathcal{V}} \left( b(v) \right) ^ 2$, and at a NE: $J_3(a) = \sum_{i \in \mathcal{I}} \Big[ \left( b(v_i ^ *) - \sum_{i \in \mathcal{S}_{v_i ^ *}} \pi_i \right) ^ 2 - \left( b(v_i ^ *) - \pi_i ^ * - \sum_{k \in \mathcal{S}_{v_i ^ *} \setminus \{i\}} \pi_k \right) ^ 2 - \frac{\bar{\theta}_i}{\theta_i + \eta_i \pi_i ^ *} \Big]$. Then, we have
        \begin{multline}\label{eqn:lemma1_0}
            \sum_{i \in \mathcal{I}} \Bigg[ \left( b(v_i ^ *) - \sum_{k \in \mathcal{S}_{v_i ^ *} \setminus \{i\}} \pi_k \right) ^ 2 \\
            - \left( b(v_i ^ *) - \pi_i ^ * - \sum_{k \in \mathcal{S}_{v_i ^ *} \setminus \{i\}} \pi_k \right) ^ 2 - \frac{\bar{\theta}_i}{\theta_i + \eta_i \pi_i ^ *} \Bigg] \\
            \leq J_3(a ^ *) - \sqrt{(\tilde{b} ^ 2 + 2 (J_3(a ^ {\text{Nash}}) - I \bar{\theta}_i)) (\tilde{b} ^ 2 + 2 ( J_3(a ^ *) - I \bar{\theta}_i))} \\
            - 4 I \bar{\theta}_i - \tilde{b} ^ 2 \\
            - \tilde{b} \left( \sqrt{\tilde{b} ^ 2 + 2 (J_3(a ^ {\text{Nash}}) - I \bar{\theta}_i)} + \sqrt{\tilde{b} ^ 2  + 2 (J_3(a ^ *) - I \bar{\theta}_i)} \right).
        \end{multline}
\end{lemma}

\begin{proof}
    At social optimum, the pricing mechanism is given by
        \begin{multline}
            \tau ^ *(v_i ^ *, \pi_i ^ *) = \left( b(v_i ^ *) - \sum_{k \in \mathcal{S}_{v_i ^ *} ^ * \setminus \{i\}} \pi_k ^ * \right) ^ 2 - \left( b(v_i ^ *) - \sum_{i \in \mathcal{S}_{v_i ^ *} ^ *} \pi_i ^ * \right) ^ 2 \\
            = \left( b(v_i ^ *) + \sum_{k \in \mathcal{S}_{v_i ^ *} ^ * \setminus\{i\}} \pi_k ^ * \right) ^ 2 - \left( b(v_i ^ *) - \pi_i ^ * - \sum_{k \in \mathcal{S}_{v_i ^ *} ^ * \setminus \{i\}} \pi_k ^ * \right) ^ 2 \\
            = 2 \pi_i ^ * b(v_i ^ *) - (\pi_i ^ *) ^ 2 - 2 \pi_i ^ * \sum_{k \in \mathcal{S}_{v_i ^ *} ^ * \setminus\{i\}} \pi_k ^ *.
        \end{multline}
    Summing over all travelers now gives
        \begin{align}
            & \sum_{i \in \mathcal{I}} \left[ 2 \pi_i ^ * b(v_i ^ *) - (\pi_i ^ *) ^ 2 - 2 \pi_i ^ * \sum_{k \in \mathcal{S}_{v_i ^ *} ^ * \setminus\{i\}} \pi_k ^ * \right] \notag \\
            & = 2 \sum_{i \in \mathcal{I}} \pi_i ^ * b(v_i ^ *) - \sum_{v \in \mathcal{V}} \left( 2 \left( \sum_{i \in \mathcal{S}_{v ^ *} ^ *} \pi_i ^ * \right) ^ 2 \right) + \sum_{i \in \mathcal{I}} (\pi_i ^ *) ^ 2 \notag \\
            & = 2 \sum_{v \in \mathcal{V}} b(v) \sum_{i \in \mathcal{S}_{v_i ^ *} ^ *} \pi_i ^ * - \sum_{v \in \mathcal{V}} \left( 2 \left( \sum_{i \in \mathcal{S}_{v ^ *} ^ *} \pi_i ^ * \right) ^ 2 \right) + \sum_{i \in \mathcal{I}} (\pi_i ^ *) ^ 2. \label{eqn:lemma1_1}
        \end{align}
    So, we use the Cauchy-Schwarz inequality to bound \eqref{eqn:lemma1_1}, i.e.,
        \begin{multline}\label{eqn:lemma1_15}
            J_3 (a ^ *) - \sum_{i \in \mathcal{I}} \frac{\bar{\theta}_i}{\theta_i + \eta_i \pi_i ^ *} \leq 2 \sqrt{\sum_{v \in \mathcal{V}} \left( b(v) \right) ^ 2 \sum_{v \in \mathcal{V}} \left( \sum_{i \in \mathcal{S}_{v ^ *} ^ *} \pi_i ^ * \right) ^ 2} \\
            - 2 \sum_{v \in \mathcal{V}} \left( \sum_{i \in \mathcal{S}_{v ^ *} ^ *} \pi_i ^ * \right) ^ 2 + \sum_{i \in \mathcal{I}} (\pi_i ^ *) ^ 2.
        \end{multline}
    For any traveler $i$, it is always true that $\frac{\bar{\theta}_i}{\theta_i + \eta_i \pi_i ^ *} > 0$, $\pi_i ^ * \in [0, \bar{\theta}_i]$, and also $\tilde{b} ^ 2 = \sum_{v \in \mathcal{V}} ( b(v) ) ^ 2$. Thus, \eqref{eqn:lemma1_15} simplifies to
        \begin{multline}\label{eqn:lemma1_2}
            \sum_{v \in \mathcal{V}} \left( \sum_{i \in \mathcal{S}_{v ^ *} ^ *} \pi_i ^ * \right) ^ 2 - \tilde{b} \sqrt{\sum_{v \in \mathcal{V}} \left( \sum_{i \in \mathcal{S}_{v ^ *} ^ *} \pi_i ^ * \right) ^ 2} \\
            - \frac{1}{2} \left( J_3(a ^ *) - I \cdot \bar{\theta}_i \right) \leq 0,
        \end{multline}
    Note that \eqref{eqn:lemma1_2} is a second-order polynomial with respect to $\sqrt{\sum_{v \in \mathcal{V}} \left( \sum_{i \in \mathcal{S}_{v ^ *} ^ *} \pi_i ^ * \right) ^ 2}$. Thus, we compute the discriminant $\Delta ^ * = \tilde{b} ^ 2 + 2 \left( J_3(a ^ *) - I \cdot \bar{\theta}_i \right)$, where $\Delta ^ *$ denotes the discriminant at the social optimum, so clearly $\Delta ^ * \geq 0$. So, from \eqref{eqn:lemma1_2}, we get
        \begin{equation}\label{eq:OPT_delta}
            2 \sqrt{\sum_{v \in \mathcal{V}} \left( \sum_{i \in \mathcal{S}_{v ^ *} ^ *} \pi_i ^ * \right) ^ 2} \leq \tilde{b} + \sqrt{\Delta ^ *}.
        \end{equation}
    Now, we need to follow the same steps to obtain a similar relation as \eqref{eq:OPT_delta} for a NE. Hence, we have
        \begin{equation}\label{eqn:NE_delta}
            2 \sqrt{\sum_{v \in \mathcal{V}} \left( \sum_{k \in \mathcal{S}_{v}} \pi_k \right) ^ 2} \leq \tilde{b} + \sqrt{\Delta},
        \end{equation}
    where $\Delta = \tilde{b} ^ 2 + 2 \left( J_3(a ^ {\text{Nash}}) - I \cdot \bar{\theta}_i \right)$. Next, the LHS of \eqref{eqn:lemma1_0}, if expanded, can be simplified as follows:
        \begin{multline}
            \sum_{i \in \mathcal{I}} \Bigg[ \left( b(v_i ^ *) - \sum_{k \in \mathcal{S}_{v_i ^ *} \setminus\{i\}} \pi_k \right) ^ 2 \\
            - \left( b(v_i ^ *) - \pi_i ^ * - \sum_{k \in \mathcal{S}_{v_i ^ *} \setminus\{i\}} \pi_k \right) ^ 2 - \frac{\bar{\theta}_i}{\theta_i + \eta_i \pi_i ^ *} \Bigg] \\
            = 2 \sum_{i \in \mathcal{I}} \left( b(v_i ^ *) \pi_i ^ * - \pi_i ^ * \sum_{k \in \mathcal{S}_{v_i ^ *} \setminus\{i\}} \pi_k - \frac{1}{2} (\pi_i ^ *) ^ 2 \right) \\
            - \sum_{i \in \mathcal{I}} \frac{\bar{\theta}_i}{\theta_i + \eta_i \pi_i ^ *}.
        \end{multline}
        \begin{multline}
            J_3(a ^ *) - 2 \sum_{i \in \mathcal{I}} \pi_i ^ * \sum_{k \in \mathcal{S}_{v_i ^ *} \setminus\{i\}} \pi_k + 2 \sum_{i \in \mathcal{I}} \pi_i ^ * \sum_{k \in \mathcal{S}_{v_i ^ *} ^ * \setminus\{i\}} \pi_k ^ * \\
            = J_3(a ^ *) - 2 \sum_{i \in \mathcal{I}} \pi_i ^ * \left[ \sum_{k \in \mathcal{S}_{v_i ^ *} \setminus\{i\}} \pi_k - \sum_{k \in \mathcal{S}_{v_i ^ *} ^ * \setminus\{i\}} \pi_k ^ * \right] \\
            = J_3(a ^ *) - 2 \sum_{i \in \mathcal{I}} \pi_i ^ * \left( \sum_{k \in \mathcal{S}_{v_i ^ *}} \pi_k ^ * - \pi_i(v_i ^ *) - \sum_{k \in \mathcal{S}_{v_i ^ *} ^ *} \pi_k ^ * + \pi_i ^ * \right), \label{eqn:lemma2_2.5}
        \end{multline}
    where $\pi_i(v_i ^ *)$ denotes traveler $i$'s payment at an optimal $v_i ^ *$, and thus, \eqref{eqn:lemma2_2.5} can be simplified by noting that
        \begin{multline}
            2 \sum_{i \in \mathcal{I}} \pi_i ^ * \left( \sum_{k \in \mathcal{S}_{v_i ^ *}} \pi_k ^ * - \pi_i(v_i ^ *) - \sum_{k \in \mathcal{S}_{v_i ^ *} ^ *} \pi_k ^ * + \pi_i ^ * \right) = \\
            2 \sum_{i \in \mathcal{I}} \left[ (\pi_i ^ *) ^ 2 - \pi_i ^ * \pi_i(v_i ^ *) \right] + 2 \sum_{v \in \mathcal{V}} \sum_{k \in \mathcal{S}_{v ^ *} ^ *} \pi_k ^ * \sum_{k \in \mathcal{S}_{v}} \pi_k \\
            - 2 \sum_{v \in \mathcal{V}} \left( \sum_{k \in \mathcal{S}_{v ^ *} ^ *} \pi_k ^ * \right) ^ 2, \label{eqn:lemma2_2.6}
        \end{multline}
    Using \eqref{eqn:lemma2_2.6}, we impose an upper bound to \eqref{eqn:lemma2_2.5} as follows
        \begin{multline}
            J_3(a ^ *) - 2 \sum_{i \in \mathcal{I}} \pi_i ^ * \left( \sum_{k \in \mathcal{S}_{v_i ^ *}} \pi_k ^ * - \pi_i(v_i ^ *) - \sum_{k \in \mathcal{S}_{v_i ^ *} ^ *} \pi_k ^ * + \pi_i ^ * \right) \\
            \leq J_3(a ^ *) - 4 I \bar{\theta}_i - 2 \sum_{v \in \mathcal{V}} \left[ \sum_{k \in \mathcal{S}_{v ^ *} ^ *} \pi_k ^ * \sum_{k \in \mathcal{S}_{v}} \pi_k \right]. \label{eqn:lemma2_3}
        \end{multline}
    We continue by upper bounding the summations in \eqref{eqn:lemma2_3}:
        \begin{align}
            & \sum_{v \in \mathcal{V}} \left[ \sum_{k \in \mathcal{S}_{v ^ *} ^ *} \pi_k ^ * \sum_{k \in \mathcal{S}_{v}} \pi_k \right] \leq \sqrt{\sum_{v \in \mathcal{V}} \left( \sum_{k \in \mathcal{S}_{v ^ *} ^ *} \pi_k ^ * \right) ^ 2 \sum_{v \in \mathcal{V}} \left( \sum_{k \in \mathcal{S}_{v}} \pi_k \right) ^ 2} \notag \\
            & \leq \frac{1}{2} \left( \sqrt{\tilde{b} ^ 2} + \sqrt{\Delta} \right) \left( \sqrt{\tilde{b} ^ 2} + \sqrt{\Delta ^ *} \right) \notag \\
            & = \frac{1}{2} \sqrt{\Delta \Delta ^ *} + \frac{\tilde{b}}{2} \left( \sqrt{\Delta} + \sqrt{\Delta ^ *} \right) + \frac{\tilde{b} ^ 2}{2}, \label{eqn:lemma1_final}
        \end{align}
    by the Cauchy-Schwartz inequality and relations \eqref{eq:OPT_delta} and \eqref{eqn:NE_delta}. Finally, we substitute $\Delta = \tilde{b} ^ 2 + 2 (J_3(a ^ {\text{Nash}}) - I \bar{\theta}_i)$ and $\Delta ^ * = \tilde{b} ^ 2 + 2 (J_3(a ^ *) - I \bar{\theta}_i)$ into \eqref{eqn:lemma1_final} and with \eqref{eqn:lemma2_3} we obtain the desired bound.
\end{proof}

\section{Proof of Lemma \ref{lemma2}}
\label{appendix_3}

\begin{lemma}\label{lemma2}
    We have
        \begin{multline}\label{eqn:lemma2-0}
            \frac{J(a ^ *)}{J(a ^ {\text{Nash}})} \leq 1 - \sqrt{ \left( \frac{\tilde{b} ^ 2}{I} + \frac{2 J(a ^ *)}{J(a ^ {\text{Nash}})} \right) \left( \frac{\tilde{b} ^ 2}{I} + 2 \right)} + \frac{\tilde{b} ^ 2}{I} \\
            - 4 \bar{\theta}_i - \frac{\tilde{b}}{\sqrt{I}} \left( \sqrt{\frac{\tilde{b} ^ 2}{I} + \frac{2 J(a ^ *)}{J(a ^ {\text{Nash}})}} + \sqrt{\frac{\tilde{b} ^ 2}{I} + 2} \right).
        \end{multline}
\end{lemma}

\begin{proof}
    We can show this result by expanding and rearranging \eqref{eqn:lemma2-0} to obtain a simplified relation. So, we have
        \begin{multline}\label{eqn:lemma2-1}
            \frac{J(a ^ *)}{J(a ^ {\text{Nash}})} - \left( 1 + \frac{\tilde{b} ^ 2}{I} - 4 \bar{\theta}_i \right) + \frac{\tilde{b}}{\sqrt{I}} \sqrt{\frac{\tilde{b} ^ 2}{I} + 2} \\
            \leq - \left( \sqrt{\frac{\tilde{b} ^ 2}{I} + 2} + \frac{\tilde{b}}{\sqrt{I}} \right) \sqrt{\frac{\tilde{b} ^ 2}{I} + \frac{2 J(a ^ *)}{J(a ^ {\text{Nash}})}}.
        \end{multline}
    We seek to solve for $\frac{J(a ^ *)}{J(a ^ {\text{Nash}})}$, so we remove the square roots by squaring both sides of \eqref{eqn:lemma2-1}, i.e.,
        \begin{equation}\label{eqn:lemma2-23}
            \frac{J(a ^ *)}{J(a ^ {\text{Nash}})} - 2(D + E ^ 2) \frac{J(a ^ *)}{J(a ^ {\text{Nash}})} + \left( D ^ 2 - \frac{\tilde{b} ^ 2 E ^ 2}{I} \right) \leq 0,
        \end{equation}
    where $E = \sqrt{\frac{\tilde{b} ^ 2}{I} + 2} + \frac{\tilde{b}}{\sqrt{I}}$ and $D = \left( 1 + \frac{\tilde{b} ^ 2}{I} - 4 \bar{\theta}_i \right) - \frac{\tilde{b}}{\sqrt{I}} \sqrt{\frac{\tilde{b} ^ 2}{I} + 2}$. We solve \eqref{eqn:lemma2-23} by noting the positivity of the coefficients to obtain
        \begin{equation}\label{eqn:lemma2-2}
            \frac{J(a ^ *)}{J(a ^ {\text{Nash}})} \leq E ^ 2 + D + E \sqrt{E ^ 2 + 2 D + \frac{\tilde{b} ^ 2}{I}}.
        \end{equation}
    We observe that $E ^ 2 + 2 D + \frac{\tilde{b} ^ 2}{I} \leq \left( E + \frac{D}{E} + \frac{\tilde{b} ^ 2}{2 E I} \right) ^ 2$, and so, an upper bound exists for \eqref{eqn:lemma2-2}. Thus, we have $\frac{J(a ^ *)}{J(a ^ {\text{Nash}})} \leq 2 E ^ 2 + 2 D + \frac{\tilde{b} ^ 2}{2 I}$. We substitute back both $E$ and $D$ and get
        \begin{align}
            \frac{J(a ^ *)}{J(a ^ {\text{Nash}})} & \leq 2 + \frac{3 \tilde{b} ^ 2}{2 I} + \frac{3 \tilde{b}}{\sqrt{I}} \sqrt{\frac{\tilde{b} ^ 2}{I} + 2} + \frac{5 \tilde{b} ^ 2}{2 I} \notag \\
            & = 2 + \frac{4 \tilde{b} ^ 2}{I} + \frac{ \tilde{b}}{\sqrt{I}} \sqrt{\frac{\tilde{b} ^ 2}{I} + 2}.
        \end{align}
    Hence, since $\frac{\tilde{b} ^ 2}{I} + 2 \leq \left( \frac{\tilde{b}}{\sqrt{I}} + \frac{\sqrt{I}}{\tilde{b}} \right) ^ 2$, the result follows.
\end{proof}

\section{Proof of Theorem \ref{thm:PoA}}
\label{appendix_4}

\begin{proof}
    From Definition \ref{defn:NE}, for some arbitrary traveler $i \in \mathcal{I}$, it is clear that $u_i(a_i ^ {\text{Nash}}, a_{- i}) \leq u_i(a_i ^ *, a_{- i})$, so if we expand the RHS of it, we get
        \begin{multline}\label{eqn:thm2_1}
            u_i(a_i ^ *, a_{- i}) = \left( b(v_i ^ *) - \sum_{k \in \mathcal{S}_{v_i ^ *} \setminus \{i\}} \pi_k \right) ^ 2 \\
            - \left( b(v_i ^ *) - \pi_i ^ * - \sum_{k \in \mathcal{S}_{v_i ^ *} \setminus \{i\}} \pi_k \right) ^ 2 - \sum_{e \in \rho_i ^ * \setminus \rho_i} c_e(J_e + 1) \\
            - \sum_{e \in \rho_i ^ * \cap \rho_i} c_e(J_e) - \frac{|\mathcal{S}_{v_i ^ *}|}{\sigma(v_i ^ *, h_i)} - \frac{\bar{\theta}_i}{\theta_i + \eta_i \pi_i ^ *},
        \end{multline}
    where $b(v_i ^ *)$ is the budget at an optimal $v_i ^ *$ and $\mathcal{S}_{v_i ^ *} = \{k \in \mathcal{I} \; | \; v_k ^ * = v_i ^ *\}$. Summing over all travelers in \eqref{eqn:thm2_1}, and keeping note of $u_i(a_i ^ {\text{Nash}}, a_{- i}) \leq u_i(a_i ^ *, a_{- i})$ yields
        \begin{equation}\label{eqn:thm_PoA_2}
            J(a ^ {\text{Nash}}) \leq \sum_{i \in \mathcal{I}} u_i(a_i ^ *, a_{- i}).
        \end{equation}
    At this point, we recall that the travel time latency functions are linear. Thus, we can write
        \begin{align}
            J_1(a ^ {\text{Nash}}) & = \sum_{i \in \mathcal{I}} \sum_{e \in \rho_i ^ {\text{Nash}}} c_e(J_e ^ {\text{Nash}}) = \sum_{e \in \mathcal{E}} J_e ^ {\text{Nash}} (\xi_1 J_e ^ {\text{Nash}} + \xi_2), \\
            J_1(a ^ *) & = \sum_{i \in \mathcal{I}} \sum_{e \in \rho_i ^ *} c_e(J_e ^ *) = \sum_{e \in \mathcal{E}} J_e ^ *(\xi_1 J_e ^ * + \xi_2),
        \end{align}    
    where $J_e ^ {\text{Nash}}$ and $J_e ^ *$ denote the number of services on road $e$ at a NE and social optimum, respectively. Inspired from \cite{Awerbuch2005}, we impose an upper bound on each component of the RHS of \eqref{eqn:thm_PoA_2}, and thus we have the following
        \begin{multline}
            \sum_{i \in \mathcal{I}} \left( \sum_{e \in \rho_i ^ * \setminus \rho_i} c_e(J_e + 1) + \sum_{e \in \rho_i ^ * \cap \rho_i} c_e(J_e) \right) \\
            \leq \sum_{i \in \mathcal{I}} \sum_{e \in \rho_i ^ *} c_e(J_e + 1) = \sum_{e \in \mathcal{E}} \xi_1 J_e ^ * J_e + J_e ^ *(\xi_1 + \xi_2) \notag
        \end{multline}
        \begin{align}
            & \leq \sqrt{\sum_{e \in \mathcal{E}} \xi_1 J_e ^ 2 \xi_1 (J_e ^ *) ^ 2} + \sum_{e \in \mathcal{E}} J_e ^ * (\xi_1 J_e ^ * + \xi_2) \notag \\
            & \leq \sqrt{\sum_{e \in \mathcal{E}} (\xi_1 J_e ^ 2 + \xi_2 J_e) (\xi_1 (J_e ^ *) ^ 2 + \xi_2 J_e ^ *)} + \sum_{e \in \mathcal{E}} J_e ^ * (\xi_1 J_e ^ * + \xi_2) \notag \\
            & = \sqrt{J_1(a ^ {\text{Nash}}) \times J_1(a ^ *)} + J_1(a ^ *), \label{eqn:thm2_J1}
        \end{align}
    where we simplified the notation as $J_e ^ {\text{Nash}} = J_e$, used $c_e(J_e) \leq c_e(J_e + 1)$ for each $e \in \rho_i ^ * \cap \rho_i$, and applied the Cauchy-Schwarz inequality twice. Note that $\xi_1 J_e ^ * + \xi_2 \leq \xi_1 (J_e ^ *) ^ 2 + \xi_2 J_e ^ *$ at any $e \in \mathcal{E}$. Next, we introduce notation: $J_2(a ^ {\text{Nash}}) = \sum_{i \in \mathcal{I}} \frac{|\mathcal{S}_{v_i} ^ {\text{Nash}}|}{\sigma(v_i ^ {\text{Nash}}, h_i)}$ and $J_2(a ^ *) = \sum_{i \in \mathcal{I}} \frac{|\mathcal{S}_{v_i ^ *}|}{\sigma(v_i ^ *, h_i)}$. We have
        \begin{equation}\label{eqn:thm2_J2}
            \sum_{i \in \mathcal{I}} \frac{|\mathcal{S}_{v_i ^ *}|}{\sigma(v_i ^ *, h_i)} \leq \sqrt{J_2(a ^ *) J_2(a ^ {\text{Nash}})} + J_2(a ^ *).
        \end{equation}
    Now we introduce the following notation:
        \begin{multline}
            J_3(a ^ {\text{Nash}}) = \sum_{i \in \mathcal{I}} \Bigg[ \left( b(v_i ^ {\text{Nash}}) - \sum_{k \in \mathcal{S}_{v_i} ^ {\text{Nash}} \setminus \{i\}} \pi_k ^ {\text{Nash}} \right) ^ 2 \\
            - \left( b(v_i ^ {\text{Nash}}) - \sum_{i \in \mathcal{S}_{v_i} ^ {\text{Nash}}} \pi_i ^ {\text{Nash}} \right) ^ 2 - \frac{\bar{\theta}_i}{\theta_i + \eta_i \pi_i ^ {\text{Nash}}} \Bigg],
        \end{multline}
    where $b(v_i ^ {\text{Nash}})$ is the budget at a $v_i ^ {\text{Nash}}$, and
        \begin{multline}\label{eqn:J3_defn}
            J_3(a ^ *) = \sum_{i \in \mathcal{I}} \Bigg[ \left( b(v_i ^ *) - \sum_{k \in \mathcal{S}_{v_i ^ *} ^ * \setminus \{i\}} \pi_k ^ * \right) ^ 2 \\
            - \left( b(v_i ^ *) - \sum_{i \in \mathcal{S}_{v_i ^ *} ^ *} \pi_i ^ * \right) ^ 2 - \frac{\bar{\theta}_i}{\theta_i + \eta_i \pi_i ^ *} \Bigg].
        \end{multline}
    By Lemma \ref{lemma1} (see Appendix \ref{appendix_2}), we have a bound for the first two components of \eqref{eqn:thm_PoA_2}, thus
        \begin{multline}\label{eqn:thm_PoA_3}
            \sum_{i \in \mathcal{I}} \Bigg[ \left( b(v_i ^ *) - \sum_{k \in \mathcal{S}_{v_i ^ *} \setminus \{i\}} \pi_k \right) ^ 2
            \\ - \left( b(v_i ^ *) - \pi_i ^ * - \sum_{k \in \mathcal{S}_{v_i ^ *} \setminus \{i\}} \pi_k \right) ^ 2 - \frac{\bar{\theta}_i}{\theta_i + \eta_i \pi_i ^ *} \Bigg] \\
            \leq J_3(a ^ *) - \sqrt{(\tilde{b} ^ 2 + 2 (J_3(a ^ {\text{Nash}}) - I \bar{\theta}_i)) (\tilde{b} ^ 2 + 2 ( J_3(a ^ *) - I \bar{\theta}_i))} \\
            - 4 I \bar{\theta}_i - \tilde{b} ^ 2 \\
            - \tilde{b} \left( \sqrt{\tilde{b} ^ 2 + 2 (J_3(a ^ {\text{Nash}}) - I \bar{\theta}_i)} + \sqrt{\tilde{b} ^ 2  + 2 (J_3(a ^ *) - I \bar{\theta}_i)} \right),
        \end{multline}
    We combine all relations together \eqref{eqn:thm2_J1}, \eqref{eqn:thm2_J2}, and \eqref{eqn:thm_PoA_3} and substitute them into \eqref{eqn:thm_PoA_2} to get
        \begin{multline}\label{eqn:thm2-eqn3}
            J(a ^ {\text{Nash}}) \leq \sqrt{J_1(a ^ {\text{Nash}}) J_1(a ^ *)} + J_1(a ^ *) \\
            + \sqrt{J_2(a ^ *) J_2(a ^ {\text{Nash}})} + J_2(a ^ *) - \tilde{b} ^ 2 + J_3(a ^ *) \\
             + \sqrt{(\tilde{b} ^ 2 + 2 (J_3(a ^ {\text{Nash}}) - I \bar{\theta}_i) (\tilde{b} ^ 2 + 2 (J_3(a ^ *) - I \bar{\theta}_i))} \\
            - \tilde{b} \left( \sqrt{\tilde{b} ^ 2 + 2 (J_3(a ^ {\text{Nash}}) - I \bar{\theta}_i)} + \sqrt{\tilde{b} ^ 2 + 2 (J_3(a ^ *) - I \bar{\theta}_i)} \right) \\
            \leq J(a ^ *) - \tilde{b} \left( \sqrt{\tilde{b} ^ 2 + 2 J_3(a ^ {\text{Nash}})} + \sqrt{\tilde{b} ^ 2 + 2 J_3(a ^ *)} \right) \\
            - \sqrt{ \left( \tilde{b} ^ 2 + 2 \left( \sum_{k = 1} ^ 3 J_k(a ^ *) \right) \right) \left( \tilde{b} ^ 2 + 2 \left( \sum_{k = 1} ^ 3 J_k(a ^ {\text{Nash}}) \right) \right)} \\
            - 4 I \bar{\theta}_i - \tilde{b} ^ 2,
        \end{multline}
    where we have used the fact that for any four numbers $(\gamma_k \in \mathbb{R}_{> 0})$, $k = 1, 2, 3, 4$, we have $\sqrt{\gamma_1 \gamma_2} + \sqrt{\gamma_3 \gamma_4} \leq \sqrt{(\gamma_1 + \gamma_3) (\gamma_2 + \gamma_4)}$. Then \eqref{eqn:thm2-eqn3} leads to
        \begin{multline}
            J(a ^ *) - \sqrt{(\tilde{b} ^ 2 + 2 J(a ^ *)) (\tilde{b} ^ 2 + 2 J(a ^ {\text{Nash}}))} \\
             - \tilde{b} \left( \sqrt{\tilde{b} ^ 2 + 2 J_3(a ^ {\text{Nash}})} + \sqrt{\tilde{b} ^ 2 + 2 J_3(a ^ *)} \right) - 4 I \bar{\theta}_i - \tilde{b} ^ 2 \\
            \leq J(a ^ *) - \sqrt{(\tilde{b} ^ 2 + 2 J(a ^ *)) (\tilde{b} ^ 2 + 2 J(a ^ {\text{Nash}}))} \\
             - \tilde{b} \left( \sqrt{\tilde{b} ^ 2 + 2 J(a ^ {\text{Nash}})} + \sqrt{\tilde{b} ^ 2 + 2 J(a ^ *)} \right) - 4 I \bar{\theta}_i - \tilde{b} ^ 2.
        \end{multline}
    So, we have the following after a simple rearrangement
        \begin{multline}\label{eqn:thm2-eqn4}
            J(a ^ *) \leq J(a ^ {\text{Nash}}) - \sqrt{(\tilde{b} ^ 2 + 2 J(a ^ *)) (\tilde{b} ^ 2 + 2 J(a ^ {\text{Nash}}))} \\
            - \tilde{b} \left( \sqrt{\tilde{b} ^ 2 + 2 J(a ^ {\text{Nash}})} + \sqrt{\tilde{b} ^ 2 + 2 J(a ^ *)} \right) - 4 I \bar{\theta}_i - \tilde{b} ^ 2.
        \end{multline}
    We divide both sides of \eqref{eqn:thm2-eqn4} by $J(a ^ {\text{Nash}})$ to obtain
        \begin{multline}
            \frac{J(a ^ *)}{J(a ^ {\text{Nash}})} \leq 1 - \sqrt{ \left( \frac{\tilde{b} ^ 2}{I} + \frac{2 J(a ^ *)}{J(a ^ {\text{Nash}})} \right) \left( \frac{\tilde{b} ^ 2}{I} + 2 \right)} + \frac{\tilde{b} ^ 2}{I} \\
            - 4 \bar{\theta}_i - \frac{\tilde{b}}{\sqrt{I}} \left( \sqrt{\frac{\tilde{b} ^ 2}{I} + \frac{2 J(a ^ *)}{J(a ^ {\text{Nash}})}} + \sqrt{\frac{\tilde{b} ^ 2}{I} + 2} \right),
        \end{multline}
    By Lemma \ref{lemma2}, we reach the desired bound.
\end{proof}

\section{Proof of Theorem \ref{thm:existence_NASH_PT}}
\label{appendix_5}

\begin{proof}
    We expand $z(n)$ and subtract $z_i ^ 0$ and simplify to get
        \begin{align}
            z(n) & = \left( n - \sum_{k \in \mathcal{S}_{v_i} \setminus \{i\}} \pi_k \right) ^ 2 - \left( n - \sum_{i \in \mathcal{S}_{v_i}} \pi_i \right) ^ 2 \notag \\
            z(n) - z_i ^ 0 & = 2 n \left[ \sum_{k \in \mathcal{S}_{v_i} \setminus \{i\}} \pi_k - \sum_{i \in \mathcal{S}_{v_i}} \pi_i \right] \notag \\
            & = 2 n \pi_i. \label{eqn:thmPT_first}
        \end{align}
    where $z_i ^ 0 =  \left( \sum_{k \in \mathcal{S}_{v_i} \setminus \{i\}} \pi_k \right) ^ 2 - \left( \sum_{i \in \mathcal{S}_{v_i}} \pi_i \right) ^ 2$. Substituting \eqref{eqn:thmPT_first} into \eqref{eqn:utility_prospect} yields
        \begin{multline}
            u_i ^ {\text{PT}}(a) = z_i ^ 0 + \sum_{n \in \mathbb{R}} \nu_i \cdot ( 2 n \pi_i) \cdot w_i(f(n)) \\
            - \sum_{e \in \rho_i : \rho_i \in \mathcal{P} ^ {(o, d)}} c_e(J_e) - \frac{|\mathcal{S}_{v_i}|}{\sigma(v_i, h_i)} - \frac{\bar{\theta}_i}{\theta_i + \eta_i \pi_i},
        \end{multline}
    where $v_i$ is given by \eqref{eqn:valuation_prospect}. The next step is to explicitly define a new potential function under prospect theory. We have
        \begin{multline}\label{eqn:thmPT_second}
            \Psi(a) = \sum_{i \in \mathcal{I}} \sum_{n \in \mathbb{R}} \nu_i \cdot ( 2 n \pi_i ) \cdot w_i(f(n)) \\
            - \sum_{e \in \mathcal{E}} \sum_{k = 1} ^ {J_e} c_e(k) - \sum_{v \in \mathcal{V}} \frac{|\mathcal{S}_{v}| (|\mathcal{S}_{v}| + 1)}{2 \sigma(v, h_i)} \\
            - \sum_{i \in \mathcal{I}} \frac{\bar{\theta}_i}{\theta_i + \eta_i \pi_i} + \sum_{v \in \mathcal{V}} \left( \sum_{k \in \mathcal{S}_{v}} \pi_k \right) ^ 2.
        \end{multline}
    Next, we show that $\Psi$ as given in \eqref{eqn:thmPT_second} is an exact potential function. We notice that $\sum_{n \in \mathbb{R}} \nu_i \cdot (2 n \pi_i) \cdot w_i(f(n))$ does not depend on $a_{- i}$, i.e., the actions of the other travelers except traveler $i$. Hence, following similar arguments as in Theorem \ref{thm:existence_NASH}, we obtain $u_i ^ {\text{PT}}(a_i, a_{- i}) - u_i ^ {\text{PT}}(a_i ', a_{- i}) = \Psi(a_i, a_{- i}) - \Psi(a_i ', a_{- i})$. Hence, $\Psi$ is indeed an exact potential function for the mobility game $\mathcal{M}$ under prospect theory. Therefore, since any action profile that minimizes $\Psi$ results in a NE, the mobility game $\mathcal{M}$ admits a NE under prospect theory.
\end{proof}

\bibliographystyle{IEEEtran}
\bibliography{references}

%
%
%
%
%

\begin{IEEEbiography}
    [{
        \includegraphics[width = 1in, height = 1.25in, clip, keepaspectratio]
        {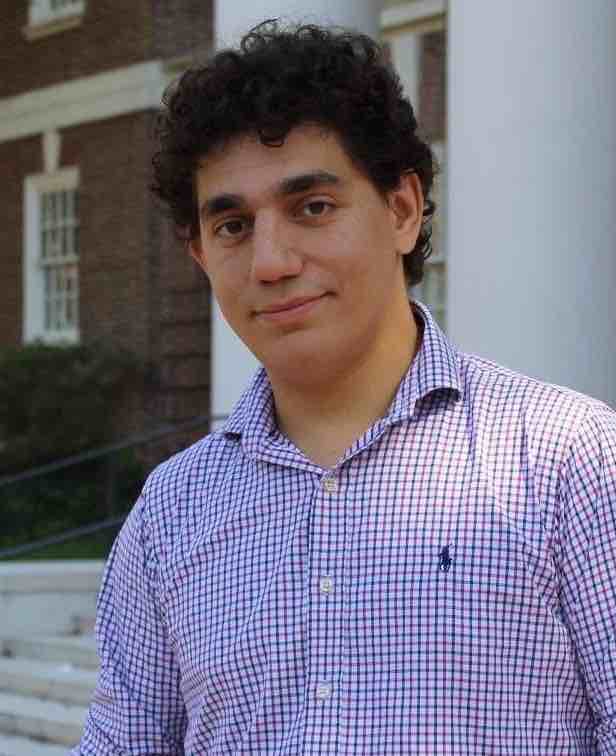}
    }]   
        {Ioannis Vasileios Chremos} (S'18) received the B.S. degree with honours of the First Class in Mathematics from the University of Glasgow, Glasgow, UK in 2017. He is working towards a Ph.D. in the program of Mechanical Engineering at the University of Delaware. His research interests lie broadly in emerging mobility systems, game theory, and mechanism design. Current research projects involve the game-theoretic study of the socioeconomic and strategic interactions of travelers in emerging mobility systems with connected and automated vehicles and the holistic development of a sociotechnical system framework for smart mobility using insights and techniques from behavioral economics and mechanism design. He is also interested in game-theoretic models of social media networks and the study of different theoretical prevention methodologies of the spread of misinformation. He is a student member of SIAM and ASME.
\end{IEEEbiography}

\begin{IEEEbiography}
    [{
        \includegraphics[width=1in,height=1.25in,clip,keepaspectratio]{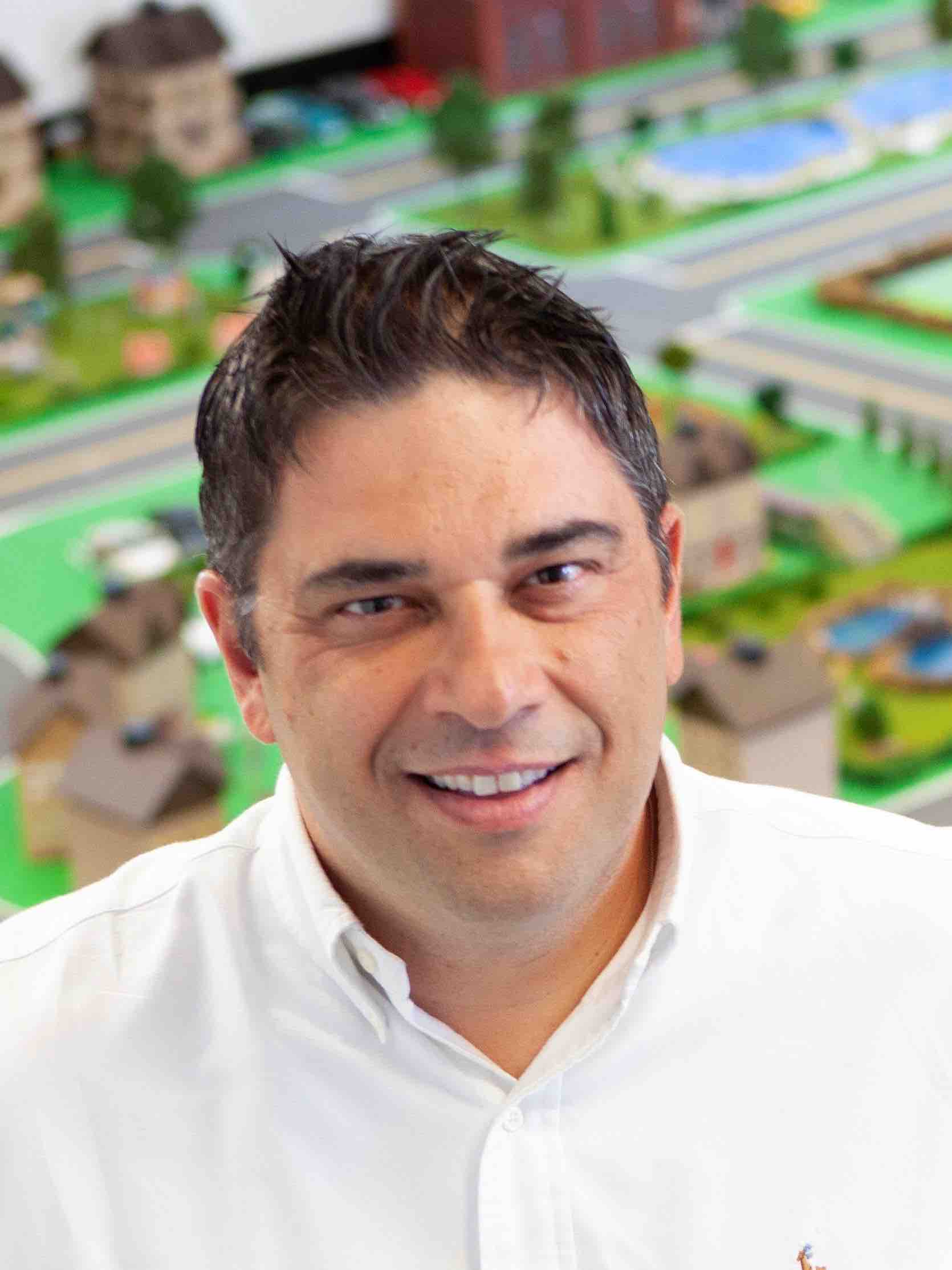}
    }]
        {Andreas A. Malikopoulos} (S'06--M'09--SM'17) received the Diploma in mechanical engineering from the National Technical University of Athens, Greece, in 2000. He received M.S. and Ph.D. degrees from the department of mechanical engineering at the University of Michigan, Ann Arbor, Michigan, USA, in 2004 and 2008, respectively. He is the Terri Connor Kelly and John Kelly Career Development Associate Professor in the Department of Mechanical Engineering at the University of Delaware, the Director of the Information and Decision Science (IDS) Laboratory, and the Director of the Sociotechnical Systems Center. Prior to these appointments, he was the Deputy Director and the Lead of the Sustainable Mobility Theme of the Urban Dynamics Institute at Oak Ridge National Laboratory, and a Senior Researcher with General Motors Global Research \& Development. His research spans several fields, including analysis, optimization, and control of cyber-physical systems; decentralized systems; stochastic scheduling and resource allocation problems; and learning in complex systems. The emphasis is on applications related to smart cities, emerging mobility systems, and sociotechnical systems. He has been an Associate Editor of the IEEE Transactions on Intelligent Vehicles and IEEE Transactions on Intelligent Transportation Systems from 2017 through 2020. He is currently an Associate Editor of Automatica and IEEE Transactions on Automatic Control. He is a member of SIAM, AAAS, and a Fellow of the ASME.
\end{IEEEbiography}

\end{document}